\documentclass{amsart}

\title[f-KT on simple Lie algebras]
{On the full Kostant-Toda lattice and the flag varieties. \\I. The singular solutions}

\author{Yuancheng Xie} 
\date{\today}
\address{Beijing International Center for Mathematical Research, Peking University,
Beijing 100871, China}
\email{xieyuancheng@bicmr.pku.edu.cn}
\subjclass[2000]{}

\usepackage{multirow}
\usepackage{url}
\usepackage{color}
\usepackage{rotating}

\countdef\x=23
\countdef\y=24
\countdef\z=25
\countdef\t=26

\def\tbox(#1,#2)#3{
\x=#1 \y=#2 
\multiply\x by 12 
\multiply\y by 12 
\z=\x \t=\y
\advance\z by 12 
\advance\t by 12 
\psline(\x,\y)(\x,\t)(\z,\t)(\z,\y)(\x,\y)
\advance\x by 6
\advance\y by 6 
\rput(\x,\y){{\bf #3}}}

\usepackage{amssymb}
\usepackage{graphicx}
\usepackage{epstopdf}
\usepackage{amscd}
\usepackage{appendix}
\usepackage{graphics}
\usepackage{MnSymbol}
\usepackage{tikz-cd}
\usepackage{framed}

\usepackage{geometry}
\usepackage{hyperref}

\def\proof{\par{\it Proof}. \ignorespaces}
\def\endproof{{\ \vbox{\hrule\hbox{%
     \vrule height1.3ex\hskip0.8ex\vrule}\hrule }}\par}

\def\sproof{\par{\it Sketch of Proof}. \ignorespaces}
\def\endsproof{{\ \vbox{\hrule\hbox{%
				\vrule height1.3ex\hskip0.8ex\vrule}\hrule }}\par}

\numberwithin{equation}{section}

\let\trueint=\int
\let\truesum=\sum
\def\int{\mathop{\textstyle\trueint}\limits}
\def\sum{\mathop{\textstyle\truesum}\limits}
\let\<=\langle
\let\>=\rangle

\def\t{\mathbf{t}}
\def\0{\mathbf{0}}

\def\edge{\ar@{-}}
\def\dedge{\ar@{.}}

\usepackage{amsmath, amsthm, amssymb, amsbsy,amsfonts,amscd}
\usepackage[mathscr]{eucal}
\usepackage{epsfig}
\newtheorem{theorem}{Theorem}[section]

\newtheorem{definition}[theorem]{Definition}
\newtheorem{proposition}[theorem]{Proposition}
\newtheorem{lemma}[theorem]{Lemma}
\newtheorem{problem}[theorem]{Problem}

\newtheorem{corollary}[theorem]{Corollary}

\newtheorem{remark}[theorem]{Remark}

\usepackage{amsfonts}
\usepackage{xy}
\usepackage{amssymb}
\usepackage{amsmath}

\DeclareMathOperator{\I}{\mathcal I}

\makeatletter
\renewcommand*\env@matrix[1][*\c@MaxMatrixCols c]{%
  \hskip -\arraycolsep
  \let\@ifnextchar\new@ifnextchar
  \array{#1}}
\makeatother

\newcommand{\thmrefer}[1]{\renewcommand\thetheorem
  {\protect\ref{#1}}\addtocounter{theorem}{-1}}

\xyoption{all}
\CompileMatrices

\begin{document}

\begin{abstract}
The full Kostant-Toda (f-KT) lattice is a natural generalization of the classical tridiagonal Toda lattice. We study singular structure of solutions of the f-KT lattices defined on simple Lie algebras in two different ways: through the $\tau$-functions and through the Kowalevski-Painlev\'e analysis. The $\tau$-function formalism relies on and is equivalent to the representation theory of the underlying Lie algebras, while the Kowalevski-Painlev\'e analysis is representation independent and we are able to characterize all the terms in the Laurent series solutions of the f-KT lattices via the structure theory of the Lie algebras. Through the above analysis we compactify the initial condition spaces of f-KT lattice by the corresponding flag varieties, that is fixing the spectral parameters which are invariant under the f-KT flows, we build a one to one correspondence between solutions of the f-KT lattices and points in the corresponding flag varieties. As all the important characters we obtain in the Kowalevski-Painlev\'e analysis are integral valued, results in this paper are valid in any field containing the rational field.
\end{abstract}

\maketitle

\tableofcontents

\section{Introduction}\label{sec:intro}

The goal of the current paper is to study singular solutions of the so-called full Kostant-Toda (f-KT) lattice which is a typical example of degenerate integrable system (a.k.a. super- or non-commutative integrable systems (c.f. \cite{Nekhoroshev1972, Gekhtman-Shapiro1999, Reshetikhin2016})) and their relations with geometry of the flag varieties. 

Toda lattice is first introduced in 1967 by Toda \cite{Toda1967a, Toda1967b} as a Hamiltonian system describing the dynamics of a chain of particles on a circle (the periodic case) or along a line (the non-periodic case) under certain exponential potentials between the nearest neighbors. In 1974, H\'enon \cite{Henon1974} and Flaschka \cite{Flaschka1974a} first recognized that the real finite periodic Toda lattice is an integrable system in the sense of Liouville, that is there exists half the dimension functionally independent conserved quantities of the phase space. Flaschka \cite{Flaschka1974a, Flaschka1974b} and Manakov \cite{Manakov1974} then found a remarkable change of variables so that Toda lattice could be put in the form of Lax pairs, and in this new form the Toda lattice describes the evolution of tridiagonal Jacobi matrices on the iso-spectral manifold. Moser \cite{Moser1975a} then constructed the general solutions for the finite real symmetric non-periodic Toda lattice via the isospectral deformation method. After that the Toda lattice model went through a series of generalizations in various directions, for example, Bogoyavlensky \cite{Bogoyavlensky1976} constructed generalized periodic Toda lattice for each simple Lie algebra. Kostant \cite{Kostant1979a}, Olshanetsky-Perelomov \cite{Olshanetsky-Perelomov1979} and Symes \cite{Symes1980} then showed that all the corresponding open Toda lattices are completely integrable systems by constructing their solutions explicitly. Turning into the new century Toda lattice is still one of the most successful models widely used to test various techniques developed in integrable systems with equally wide applications in many other branches of mathematics such as representation theory, random matrix, orthogonal polynomials, geometry, topology, and combinatorics etc.(c.f. \cite{ Kodama-Shipman2018, Deift1999, CHHL2018, Semenov-Tian-Shansky2021, Tomei2013, Kodama-Williams2015, Ercolani-Ramalheira-Tsu2022, Ercolani-Ramalheira-Tsu2023, Chernyakov-Sharygin-Talalaev2023}).

According to the theory of Kostant \cite{Kostant1979a, Kostant1979b}, Toda lattice should be viewed as the coadjoint action of a Borel subgroup $B_+$ of a simple Lie group $G$ on the dual $\mathfrak{b}_+^*$ of its Lie algebra $\mathfrak{b}_+ = \text{Lie}(B_+)$. Different realizations of $\mathfrak{b}_+^*$ correspond to several different versions of Toda lattice. For example, in $\mathfrak{sl}_n$ we can use the trace form to identify $\mathfrak{b}_+^*$ either with the linear space consisting of all symmetric matrices or with the affine space consisting of matrices in the lower Hessenberg form which corresponds to the classical QR and LU decompositions of $SL_n$ respectively. The phase spaces of the classical Toda lattice consisting of the symmetric tridiagonal matrices or the tridiagonal Hessenberg matrices can be viewed as the minimal indecomposable symplectic leaves. These systems  are called symmetric Toda lattice and Kostant-Toda lattice respectively, and they are clearly integrable with the conserved quantities given by the Chevalley invariants. The totality of the Toda flows induced by these invariants are called the Toda hierarchy. Deift, Li, Nanda and Tomei \cite{DLNT1986} considered the symmetric Toda flow on a generic orbit and showed that it is still integrable with the extra conserved quantities given by the chopping method. Ercolani, Flaschka and Singer \cite{Ercolani-Flaschka-Singer1993} then studied the Toda flows on the Hessenberg matrices and obtained similar results, later these results are generalized by Gekhtman and Shapiro \cite{Gekhtman-Shapiro1999} to Toda flows on all simple Lie algebras. Following \cite{Ercolani-Flaschka-Singer1993}, we call these systems the full symmetric Toda hierarchy and the full Kostant-Toda (f-KT) hierarchy respectively. 
According to \cite{Bloch-Gekhtman1998, Kodama-Williams2015}, there exists a map from the f-KT flows to the full symmetric Toda flows, and we will focus on the f-KT hierarchy in this paper.

As an integrable system, we insist that solutions of f-KT hierarchy should be constructed and analyzed explicitly. In this respect, Kodama-McLaughlin \cite{Kodama-McLaughlin1996} and Kodama-Ye \cite{Kodama-Ye1996} constructed the solutions explicitly for the full symmetric Toda lattice and certain generalized Toda lattice defined on $n \times n$ matrices respectively via iso-spectral deformation method. By a reduction procedure, the latter results can be used to study the f-KT flows on matrices in the Hessenberg form (which includes Lie algebras in type $A_n, B_n, C_n, G_2$ etc.). On the other hand the real regular soliton solutions of f-KT hierarchy in type $A$ and their asymptotic properties were studied in terms of the totally nonnegative parts of the flag variety in \cite{Kodama-Williams2015}.

The goal of the current paper is to understand the structures of singular solutions, i.e. solutions which blow up in finite times, of the f-KT hierarchy defined on all simple Lie algebras, and we show that the compactification of the phase spaces are exactly the flag varieties. The main techniques we use are representation theory of simple Lie algebras and the Kowalevski-Painlev\'e analysis.

With a view to studying real solutions in the future we set up the f-KT hierarchy on split simple Lie algebras over $\mathbb{R}$ in the main text. But to ease notations we only describe the main results for f-KT lattice (the first f-KT flow) over $\mathbb{C}$ in this Introduction. Let $\mathfrak{g}$ be a complex simple Lie algebra of rank $\ell$, $\mathfrak{h}$ a Cartan subalgebra of $\mathfrak{g}$ and $\Delta_+$ a chosen set of positive roots of $\mathfrak{g}$.
Let $L_{\mathfrak{g}}$ be the Lax matrix defined by
\begin{equation*}
L_{\mathfrak{g}} = \sum \limits_{i = 1}^{\ell} X_{\alpha_i} +  \sum \limits_{i = 1}^{\ell} a_i(t)H_i + \sum \limits_{\alpha \in \Delta_+}b_{\alpha}(t)Y_{\alpha},
\end{equation*}
where $\{H_i, 1 \le i \le \ell; X_{\alpha}, Y_{\alpha}, \alpha \in \Delta_+\}$ is a Chevalley basis of $\mathfrak{g}$. Let $\mathfrak{n}_{\pm} = \sum \limits_{\alpha \in \Delta_{\pm}}\mathbb{R}X_{\alpha}$ and $\mathfrak{b}_{\pm} = \mathfrak{h} + \mathfrak{n}_{\pm}$ be the corresponding nilpotent subalgebras and Borel subalgebras of $\mathfrak{g}$. Denote by $G$ the adjoint group of $\mathfrak{g}$, and $N_-, B_+$ subgroups of $G$ with Lie algebras $\mathfrak{n}_-$ and $\mathfrak{b}_+$ respectively.
Then the full Kostant-Toda lattice is defined as
\begin{equation}\label{eq:f-KThky}
\frac{d L_{\mathfrak{g}}}{d t} = [B, L_{\mathfrak{g}}], \qquad \text{with } B = \Pi_{\mathfrak{b}_+}L_{\mathfrak{g}},
\end{equation}
where $\Pi_{\mathfrak{b}_+}$ is the projection to the upper Borel $\mathfrak{b}_+$ along $\mathfrak{n}_-$. 

Let $(\rho_i, V^{\omega_i})$ be the $i$-th fundamental $\mathfrak{g}$-module with highest weight vector $v^{\omega_i}$, where $\omega_i, 1 \le i \le \ell$ are the fundamental weights. Let $(\cdot, \cdot)$ be a scalar Hermitian product on $V^{\omega_i}$ such that the weight vectors form an orthonormal basis and the operator $\rho_i(X_{\alpha})$ and $\rho_i(Y_{\alpha})$ are adjoint to each other (c.f. Section \ref{sec:representation}). Let the $j$-th $\tau$-function $\tau_j$ be defined as
\[\tau_j(t) = (v^{\omega_j}, \rho_i(\exp (tL_0))v^{\omega_i}), \qquad 1 \le i \le \ell,\]
where $L_0 = L(0)$ is the initial condition. We then have
\begin{proposition}[Corollary \ref{cor:diagonal}]\label{prop:diagonalintro}
For $t$ small enough, we have the following formula for the diagonal elements in f-KT lattice:
\begin{equation}\label{eq:aiintro}
a_i({t}) = \frac{d}{dt}\ln \tau_i({t}), \qquad 1 \le i \le \ell,
\end{equation}
and the other entries $b_{\alpha},\ \alpha \in \Delta_+$ in $L_{\mathfrak{g}}$ can be determined from $a_i(t)$ subsequently.
\end{proposition}

The proof of Proposition \ref{prop:diagonalintro} is based on the following factorization of the group element $\exp (tL_0)$
\begin{equation}\label{eq:RHintro}
\exp(tL_0) = u(t)b(t), \qquad \text{with } u(t) \in N_-,\ b(t) \in B_+,
\end{equation}
which always exists when $t$ is small enough. 
When there exists some $t_* \in \mathbb{C}$ such that the factorization \eqref{eq:RHintro} could not be performed which happens when there exists $1 \le k \le \ell$ such that $\tau_k(t_*) = 0$, we need to use the more general factorization
\[\exp (t_*L_0) = u_*\dot{w}_*b_*,\]
where $\text{id} \ne \dot{w}_* \in G$ is a representative of $w_* \in \mathcal{W}$ in the Weyl group. The set of $t_*$ where some of the $\tau$-functions vanish is called the Painlev\'e divisor. Expanding $\tau$-functions around $t_*$ by setting $t \to t + t_*$ we obtain
\begin{equation}\label{eq:tauintro}
\tau_k(t; w_*) = d_k(v^{\omega_k}, e^{tL_0}u_* \dot{w}_*v^{\omega_k}), \qquad 1 \le k \le \ell.
\end{equation}
One of the main purposes of the current paper is to identify which $u \in N_-$ and $w \in \mathcal{W}$ can appear in formula \eqref{eq:tauintro} so that $a_i(t)$'s given by \eqref{eq:aiintro} still are solutions to the f-KT lattice. The conclusion is as follows
\begin{theorem}[Theorem \ref{thm:generaltau}]\label{thm:generaltauintro}
For any $u \in N_-$ and $w \in \mathcal{W}$,
\begin{equation}\label{eq:tau2intro}
\tau_k(t; w) = (v^{\omega_k}, \exp (tL_0)u\dot{w}v^{\omega_k}), \qquad 1 \le k \le \ell,
\end{equation}
are $\tau$-functions of the f-KT lattice. The diagonal elements $a_i(t)$'s of $L_{\mathfrak{g}}$ are given by formula \eqref{eq:aiintro}.
\end{theorem}
This seemingly natural fact is somewhat nontrivial especially when the underlying field is $\mathbb{R}$ as it is known that in type $A$ there are real regular solutions of f-KT lattice with dim$\, \mathfrak{b}_+$ many free parameters which do not hit any smaller Bruhat cell (c.f. \cite{Gekhtman-Shapiro1997, Kodama-Williams2015}). We will also see that only solutions of the tridiagonal Kostant-Toda lattice will hit the smallest Bruhat cell which is another hint of the non-triviality of Theorem \ref{thm:generaltauintro}. 

We then use Kowalevski-Painlev\'e analysis to study these singular solutions (c.f. \cite{Flaschka-Haine1991} for a discussion on the tridiagonal case).

The Kowalevski-Painlev\'e analysis for nonlinear ordinary (partial) differential equations is an analogue of the classical Frobenius method for ordinary differential equations near a regular singular point (c.f. \cite{Whittaker-Watson2020}). The idea is quite simple: we prepare some formal Laurent series as candidate solutions and substitute them into the differential equations under investigation, then we recursively find the coefficients and show that the Laurent series converge and are indeed meromorphic solutions. For example, we may assume solutions of \eqref{eq:f-KThky} have the following form
\[a_i(t) = \sum \limits_{k = 0}^{\infty}a_{ik}t^{\delta_i + k} \quad 1 \le i \le \ell, \qquad b_{\alpha}(t) = \sum \limits_{k = 0}^{\infty}b_{\alpha k}t^{\gamma_{\alpha} + k} \quad \alpha \in \Delta_+. \]
In practice, the Kowalevski-Painlev\'e analysis takes the following three steps: 
\begin{enumerate}
\item[(1)] Identify the leading singularities and leading coefficients, i.e. $\delta_i, \gamma_{\alpha}$ and $a_{i0}, b_{\alpha 0}$ in the above example. 
\item[(2)] Find the resonances. Substituting the resulting form of Laurent series in step (1) into equation \eqref{eq:f-KThky} and find the coefficients recursively, the resonances are the the places where the iteration procedure to uniquely solve the higher order coefficients fail and new free parameters to be introduced. 
\item[(3)] Check the compatibility and convergence of the Laurent series. Express the other coefficients in the Laurent series in terms of the free parameters in step (2) and check whether or not the compatibility conditions are satisfied at the resonance levels. Show that the resulting formal Laurent series have a positive convergent radius.
\end{enumerate}
This method was first used by Kowalevski in her study of spinning tops \cite{Kowalevski1889} where she proposed the famous Kowalevski-Painlev\'e criterion\footnote{Roughly this criterion says if a system on the complex domain is integrable, then it should have enough singular solutions with enough free parameters (c.f. \cite{Adler-vanMoerbeke-Vanhaecke2013} for a precise statement).} for integrability and found an integrable case known as Kowalevski top nowadays. Our work is motivated directly by the results in \cite{Flaschka1988, Flaschka-Haine1991} on non-periodic tridiagonal Toda lattice and in \cite{Adler-vanMoerbeke1991} on periodic tridiagonal Toda lattice where Kowalevski-Painlev\'e analysis was used to understand the finer geometry of the level set varieties (Abelian varieties in the latter case). We show that all the three steps have appealing answers.

For linear systems, the leading powers $\delta_i, \gamma_{\alpha}$ are determined by the so-called indicial equations. But for non-linear systems such as f-KT lattice, we need to identify them by balancing the minimal power terms, we have
\begin{proposition}[Proposition \ref{prop:Uniqueness}, \cite{Singer1991}]\label{prop:Uniquenessintro}
The Laurent series solutions of f-KT lattice have the following form
\[a_{i}(t) = \frac{1}{t}\sum \limits_{k = 0}^{\infty} a_{ik}t^k \quad 1 \le i \le \ell, \qquad b_{\alpha}(t) = \frac{1}{t^{\nu_{\alpha}}} \sum \limits_{k = 0}^{\infty}b_{\alpha k}t^k \quad \alpha \in \Delta_+,\]
where $\nu_{\alpha} = L(\alpha) + 1$ and $L(\alpha)$ is the height of the root $\alpha$. That is $a_i(t), 1 \le i \le \ell$ have at most simple pole singularities.
\end{proposition}

We then substituting these Laurent series into the f-KT lattice \eqref{eq:f-KThky}, multiplying out and equating to zero the coefficients of successive powers of $t$. The first system, that is the nonlinear algebraic equations satisfied by $a_{i0}$'s and $b_{\alpha 0}$'s are called the indicial equations. Note that $b_{\alpha 0}$ can be obtained by solving linear equations once $a_{i0}$'s are known. Let $\Phi_w^{\pm}$\footnote{Sets like $\Phi_w^{\pm}$ were studied by Papi (c.f. \cite{Papi1994}).} be the set of positive (negative) coroots of $\mathfrak{g}$ which are mapped into the set of negative (positive) ones by $w$. Then we have
\begin{theorem}[Proposition \ref{conj:Nonnegative}, Theorem \ref{thm:Onetoone}]\label{thm:Onetooneintro}
There is a one-to-one correspondence between solutions of the indicial equations and elements of the Weyl group $\mathcal{W}$ of $\mathfrak{g}$. Let $w \in \mathcal{W}$, then the solution of indicial equations associated with $w$ is given by the following formula
\[a^w_{i0} = \sum \limits_{\check{\alpha} \in \Phi_w^-}\kappa( \omega_i, w^{-1}\check{\alpha}),\]
where $\omega_i$ is the $i$th fundamental weight, and $\kappa(\cdot, \cdot)$ is the bilinear form on $\mathfrak{h}^*$ induced by the restriction of the Killing form to $\mathfrak{h}$.
\end{theorem}
In plain language, this theorem says that the coefficients of the sum of coroots in $\Phi_w^+$ in the basis of simple coroots gives a solution to the indicial equations and all solutions of indicial equations can be produced in this way. This finishes step (1) in our Kowalevski-Painlev\'e analysis.

The connection between indicial equations of tridiagonal Toda lattice and the coroot system was observed before for example in \cite{Casian-Kodama2006, Flaschka-Haine1991}. Since the coroot system of $\mathfrak{g}$ is the root system of its Langlands dual $\check{\mathfrak{g}}$, the results here may have interesting applications in other branches of mathematics which we would like to explore in the future.

The other coefficients in the formal Laurent series satisfy the following recursive relations
\begin{align}\label{eq:higerordertermintro}
(kI - \mathcal{K}_w) \begin{pmatrix}a_{1k}\\ \vdots \\ a_{\ell k}\\b_{\alpha_1k} \\ \vdots \end{pmatrix} = \vec{R}_{(k)}(a_{1l}, \cdots, a_{nl}, b_{\alpha_1 l}, \cdots) \quad \text{with } l < k.
\end{align}
where $I$ is the identity matrix, $\mathcal{K}_w$ depending on solutions of the indicial equations (and so depends on some $w \in \mathcal{W}$ by Theorem \ref{thm:Onetooneintro}) is the so-called Kowalevski matrix and $\vec{R}_{(k)}$ depends only polynomially on items determined by previous steps. It should be noted that whenever $\mathcal{K}_w$ has a positive integer eigenvalue $k$, then $\det(kI - \mathcal{K}_w) = 0$ and the system \eqref{eq:higerordertermintro} is not uniquely solvable. In such a case, we may take one  (or several  depending on the multiplicity of $k$) of the $a_{ik}, b_{\alpha k}$ to be a free parameter to uniquely determine the other ones. Thus positive integer eigenvalues of $\mathcal{K}_{w}$ are the resonances in step (2). 

There are two types of eigenvalues for $\mathcal{K}_{w}$: $\ell$ of them are constants given by degrees of the Chevalley invariants, and we call them type ``C'' eigenvalues; the others depending on solutions of the indicial equations are called type ``X'' eigenvalues. As it turns out that we can associate each type-X eigenvalue of $\mathcal{K}_{w}$ with a positive root $\alpha \in \Delta_+$ of $\mathfrak{g}$, and we denote such an eigenvalue by $E^w_{\alpha}$. We then have
\begin{theorem}[Theorem \ref{thm:Kmatrix}]\label{thm:Kmatrixintro}
\[E^w_{\alpha} = L(w\alpha),\]
where $L(\alpha)$ is the height of root $\alpha \in \Delta_+$.
\end{theorem}

Theorem \ref{thm:Kmatrixintro} tells us that there is a Weyl group action acting on $E^w_{\alpha}$ which is equivariant to the Weyl group action on the root system of $\mathfrak{g}$. This gives us precise information on all eigenvalues of $\mathcal{K}_w$. In particular, we known that $\mathcal{K}_w$ has $\ell + l(w_0) - l(w)$ many positive eigenvalues, where $w_0 \in \mathcal{W}$ is the longest Weyl group element and $l(w)$ is the length of $w$. This finishes step (2).

Now we recall that the Fredholm alternative theorem in linear algebra states that: 

\begin{framed}
\begin{minipage}{0.9\linewidth}\label{cond:Fredholm}
The linear system $A\mathbf{x} = \mathbf{b}$ has a solution iff $\mathbf{y}^T\mathbf{b} = 0$ for every column vector $\mathbf{y}$ such that $\mathbf{y}^TA = 0$. 
\end{minipage}
\end{framed}

Applying this theorem to system \eqref{eq:higerordertermintro} we see that at the non-resonant level, the condition is vacuous; while at a resonant level $k$, this may impose nontrivial conditions on $\vec{R}_{(k)}$ which leads to decreasing of the number of free parameters. This is the compatibility condition we need to check in step (3). We show that all the compatibility conditions for Laurent series solutions of f-KT lattice are automatically satisfied. This is intrinsically related with the underlying Lie algebra structure. As the formal Laurent series solutions we found in Proposition \ref{prop:Uniquenessintro} are all weight homogeneous, it can be proved by majorant method that they are convergent Laurent series (c.f. \cite{Adler-vanMoerbeke-Vanhaecke2013}). This finishes step (3) for our Kowalevski-Painlev\'e analysis.

After matching all the parameters appearing in the Laurent expansion of $\tau$-functions in \eqref{eq:tau2intro} with the free parameters in the Laurent series from the Kowalevski-Painlev\'e analysis, we finally conclude Theorem \ref{thm:generaltauintro}.

Equivalently, we have
\begin{theorem}[Theorem \ref{conj:Flag}]\label{conj:Flagintro}
All the Laurent series solutions of f-KT lattice are parameterized by ${{G} \slash {B}_+ \times \mathbb{C}^{\ell}}$, where ${G} \slash {B}_+$ is the flag variety and ${ \mathbb{C}^{\ell}}$ parametrizes the data for spectral parameters. 
\end{theorem}

We remark that fixing the spectral parameters, Theorem \ref{conj:Flagintro} is saying that flag varieties are the initial condition spaces of the f-KT lattices which is the analogue of Okamoto's initial condition spaces for Painlev\'e transcendents (c.f. \cite{Okamoto1979}). Note also that as all the important characters we obtained in the Kowalevski-Painlev\'e analysis are integral valued, results in this paper are actually valid in any fields containing the rational field. 

The rest of the paper is organized as follows. In Section \ref{sec:taufunction} we introduce $\tau$-functions for f-KT hierarchy. In Section \ref{sec:Kowalevski-Painleve} we carry out complete Kowalevski-Painlev\'e analysis for the f-KT lattice. In Section \ref{sec:Painleveso5} we use f-KT lattice on rank $2$ type $B$ Lie algebra as an example to illustrate the main results. In Section \ref{sec:problems} we formulate several related problems to be considered in the future. 

\bigskip
\noindent
{\bf Acknowledgments}
The paper is partly based on the author's Ph.D. thesis \cite{Xie2021}, and the author would like to thank Yuji Kodama for his guidance and many useful suggestions. The author would also like to thank Yu Li for many useful discussions. This work is partially supported by the National Key Research and Development Program of China (No. 2021YFA1002000), by the National Natural Science Foundation of China under the Grant No. 12301304 and by the Boya Postdoctoral Fellowship of Peking University.


\section{$\tau$-functions of f-KT hierarchy}\label{sec:taufunction}
\subsection{The f-KT hierarchy}\label{sec:f-KT}
Let $\mathfrak{g}$ be a split simple Lie algebra of rank $\ell$ over $\mathbb{R}$, that is there exists a splitting Cartan subalgebra $\mathfrak{h} \subset \mathfrak{g}$ such that for all $X \in \mathfrak{h}, \text{ad}_{\mathfrak{g}} X$ is triangularizable. Choose a set of simple roots $\Pi = \{\alpha_1, \alpha_2, \dots, \alpha_{\ell}\}$ for $(\mathfrak{g}, \mathfrak{h})$ and denote by ${\Delta}_{\pm}$ the sets of positive and negative roots respectively and $\Delta = \Delta_+ \cup \Delta_-$. Let $\{ X_{\alpha}, {\alpha \in \Delta}; H_i, 1 \le i \le \ell \}$ be a Chevalley basis\footnote{Here any Cartan-Weyl basis would work, and the choice of Chevalley basis is not necessary but convenient.} of $\mathfrak{g}$. More precisely, let $H_{\alpha} \in \mathfrak{h}$ be the unique element such that $\alpha(H_{\alpha}) = 2$, and (c.f. \cite{Bourbaki2008})
\begin{enumerate}
\item $[H_{\alpha_i}, H_{\alpha_j}] = 0,\ 1 \le i, j \le \ell$.
\item $[H_{\alpha_i}, X_{\alpha}] = \langle \alpha, \alpha_i\rangle X_{\alpha}, \ 1 \le i \le \ell, \alpha \in \Delta$, where $\langle \beta, \alpha\rangle = {2\kappa(\beta, \alpha)}\slash{\kappa(\alpha, \alpha)}$ and $\kappa(\cdot, \cdot)$ is the non-degenerate symmetric bilinear form on $\mathfrak{h}^*$ induced by the restriction of the Killing form $\kappa(\cdot, \cdot)$ to $\mathfrak{h}$.
\item $[X_{\alpha}, X_{-\alpha}] = H_{\alpha}$ is a $\mathbb{Z}$-linear combination of $H_{\alpha_1}, \dots, H_{\alpha_\ell}$.
\item If $\alpha, \beta$ are independent roots, $\beta - q\alpha, \dots, \beta + p\alpha$ the $\alpha$-string through $\beta$, then $[X_{\alpha}, X_{\beta}] = 0$ if $p = 0$, while $[X_{\alpha}, X_{\beta}] = \pm (q+1)X_{\alpha + \beta}$ if $\alpha + \beta \in \Delta$.
\end{enumerate}
Here $C_{ij} = \langle \alpha_j, \alpha_i\rangle$ is the Cartan matrix of $\mathfrak{g}$. We sometimes denote $H_i := H_{\alpha_i}, 1 \le i \le \ell$ and $Y_{\alpha} := X_{-\alpha}$ for $\alpha \in \Delta_+$.

Let $\mathfrak{n}_{\pm} = \sum \limits_{\alpha \in \Delta_{\pm}}\mathbb{R}X_{\alpha}$ and $\mathfrak{b}_{\pm} = \mathfrak{h} + \mathfrak{n}_{\pm}$ be the corresponding nilpotent subalgebras and Borel subalgebras of $\mathfrak{g}$. Then $\mathfrak{g}$ admits the following decomposition,
\begin{align*}
\mathfrak{g} = \mathfrak{n}_- \oplus \mathfrak{h} \oplus \mathfrak{n}_+ = \mathfrak{n}_- \oplus \mathfrak{b}_+.
\end{align*}

Let $G^{\text{ad}}$ be the adjoint group of $\mathfrak{g}$ which in our case is just the connected component of the identity in Aut$(\mathfrak{g})$. The action of $G^{\text{ad}}$ on $\mathfrak{g}$ induces an action of $G^{\text{ad}}$ on $S(\mathfrak{g})$, the symmetric algebra over $\mathfrak{g}$. By Chevalley's theorem one knows that there exist algebraically independent homogeneous elements $I_j \in S(\mathfrak{g})^{G^{\text{ad}}}, j = 1, 2, \dots, \ell$, referred to as Chevalley invariants, which generates $S(g)^{G^{\text{ad}}}$. Let $m_j = \text{deg}(I_j) - 1, 1 \le j \le \ell$ which are sometimes called the Weyl exponents.

The f-KT hierarchy is defined as follows:
Let $L_{\mathfrak{g}}$ be the Lax matrix given by
\begin{equation}\label{eq:Lax}
L_{\mathfrak{g}}=\sum_{i=1}^{\ell} X_i + \sum_{i=1}^{\ell} a_i(\mathbf{t})H_i + \sum_{\alpha\in\Delta_+}b_{\alpha}(\mathbf{t})Y_\alpha,
\end{equation}
where  $a_i(\mathbf{t})$ and $b_{\alpha}(\mathbf{t})$ are functions of the multi-time variables $\mathbf{t}=(t_{m_k}:k=1,2,\ldots, \ell)$. 
For each time variable, we have the generalized Toda hierarchy defined as
\begin{equation}\label{eq:fKT}
\frac{\partial L_{\mathfrak{g}}}{\partial t_{m_k}}=[B_k, L_{\mathfrak{g}}],\qquad \text{with}\quad B_k=\Pi_{\mathfrak{b}_+}\nabla I_{k},\end{equation}
where  $\nabla$ is the gradient with respect to the Killing form, i.e. for any $x\in \mathfrak{g}$, $dI_k(x)=\kappa(\nabla I_k,x)$,
and $\Pi_{\mathfrak{b}_+}$ represents the $\mathfrak{b}_+$-projection.
Here the Chevalley invariants $I_k=I_k(L_{\mathfrak{g}})$ for example in type $A$ can be taken as
\[
I_{k}(L_{\mathfrak{g}})=\frac{1}{k+1}\text{tr}(L_{\mathfrak{g}}^{k+1}),\quad\text{which gives}\quad \nabla I_{k}=L_{\mathfrak{g}}^k.
\]

The first equation in \eqref{eq:Lax} is called the f-KT lattice. In this case we can take $B = B_1 = \Pi_{\mathfrak{b}_+}L_{\mathfrak{g}}$ in \eqref{eq:Lax} and write $t = t_1$. The theories we present here apply to f-KT hierarchy, but to make it concise we will use f-KT lattice to demonstrate some of the main results. In the following we will use $\t$ and $t$ to indicate results for f-KT hierarchy and f-KT lattice respectively. 


\subsection{Some background on Lie theory}
\subsubsection{Structure theory}\label{sec:structure}
We denote by $\mathfrak{g}_{\mathbb{C}}$ the complexification of $\mathfrak{g}$ in the following. Let $\mathfrak{g} = \mathfrak{k} + \mathfrak{p}$ be a Cartan decomposition of $\mathfrak{g}$ with respect to the Killing form $\kappa$ so that $\mathfrak{k}$ is the Lie algebra of a maximal compact subgroup of $G^{\text{ad}}$ and $\mathfrak{p}$ is the orthogonal complement to $\mathfrak{k}$ with respect to Re$(\kappa)$. Let $\theta$ be the corresponding Cartan involution so that $\theta =  1$ on $\mathfrak{k}$ and $\theta = -1$ on $\mathfrak{p}$. Then a compact form of $\mathfrak{g}_{\mathbb{C}}$ can be taken as $\mathfrak{u} = \mathfrak{k} + i\mathfrak{p}$. In terms of the Chevalley basis taken in Section \ref{sec:f-KT}, $\mathfrak{u}$ can be written as
\[\mathfrak{u} = \sum \limits_{1 \le i \le \ell}\mathbb{R}(iH_{i}) + \sum \limits_{\alpha \in \Delta}\mathbb{R}(X_{\alpha} - X_{-\alpha}) + \sum \limits_{\alpha \in \Delta}\mathbb{R}i(X_{\alpha} + X_{-\alpha}).\]
Note that $\mathfrak{g}_{\mathbb{C}} = \mathfrak{g} \oplus i\mathfrak{g} = \mathfrak{u} \oplus i\mathfrak{u}$, i.e. both $\mathfrak{g}$ and $\mathfrak{u}$ are real forms of $\mathfrak{g}_{\mathbb{C}}$.

Note that the real form $\mathfrak{g}$ of $\mathfrak{g}_{\mathbb{C}}$ defines a conjugation operation on $\mathfrak{g}_{\mathbb{C}}$. That is if $Z = X + iY \in \mathfrak{g}_{\mathbb{C}}$ for $X, Y \in \mathfrak{g}$, then we put $Z^c = X - iY$. This operation induces an automorphism $g \mapsto g^c$ of the adjoint group $G_{\mathbb{C}}^{\text{ad}}$ such that for any $X \in \mathfrak{g}_{\mathbb{C}}$ we have $(\exp X)^c = \exp X^c$. 
Let
\begin{equation}\label{eq:Gtilde}
\tilde{G} = \{g \in G_{\mathbb{C}}^{\text{ad}} \ \vert \ g^c = g\}.
\end{equation}
Then $\tilde{G}$ is the set of all elements $g \in G_{\mathbb{C}}^{\text{ad}}$ which stabilize $\mathfrak{g}$, and $G \subseteq \tilde{G}$. Let $\tilde{M}$ be the set all $a \in H^{\text{ad}}_{\mathbb{C}}$ such that $a^2 = 1$, then
\[\tilde{G} \cap H_{\mathbb{C}}^{\text{ad}} = \tilde{M}H^{\text{ad}}.\]
Let $M = \tilde{M} \cap G$, then $M = K \cap T$ where $K$ and $T$ are subgroups of $G^{\text{ad}}_{\mathbb{C}}$ corresponding to $\mathfrak{k}$ and $\mathfrak{t} := i\mathfrak{h}$, respectively. 

\begin{remark}\label{rem:mcomponents}
$\tilde{M}$ and $M$ measure disconnectness of the corresponding real Cartan subgroups and play important roles in the structure and representation theory of real Lie groups (c.f. \cite{Knapp2002}), in Whittaker theory (c.f. \cite{Kostant1978}), and in the study of topology of real isospectral manifolds associated with blowups of Toda lattices (c.f. \cite{Kodama-Ye1996b, Casian-Kodama2002}) etc. For our discussion of singular structure of f-KT lattice, we will mainly focus on behavior of f-KT lattice on the identity component in the following while leaving the exploration of the other connections mentioned above elsewhere.
\end{remark}

Let $G_{\mathbb{C}}^{\text{s}}$ be a fixed simply connected Lie group having $\mathfrak{g}_{\mathbb{C}}$ as its Lie algebra. 
The adjoint representation defines a homomorphism $\text{Ad}: G^{\text{s}}_{\mathbb{C}} \to G^{\text{ad}}_{\mathbb{C}}$, in fact an exact sequence:
\[1 \to \mathcal{Z}(G^{\text{s}}_{\mathbb{C}}) \to G^{\text{s}}_{\mathbb{C}} \xrightarrow{\text{Ad}} G^{\text{ad}}_{\mathbb{C}} \to 1,\]
where $\mathcal{Z}(G^{\text{s}}_{\mathbb{C}})$ is the center of $G^{\text{s}}_{\mathbb{C}}$. 

Let $U^{\text{s}}$ and $G^{\text{s}}$ be subgroups of $G_{\mathbb{C}}^{\text{s}}$ with Lie algebras $\mathfrak{u}$ and $\mathfrak{g}$ respectively. 
Let $N^{\text{s}}_-, H^{\text{s}}, N^{\text{s}}_+$ be the subgroups of $G^{\text{s}}$ and $N^{\text{ad}}_-, H^{\text{ad}}, N^{\text{ad}}_+$ be the subgroups of $G^{\text{ad}}$ associated with $\mathfrak{n}_-, \mathfrak{h}$ and $\mathfrak{n}_+$, respectively. 
Let $G^{\text{s}}_o = N_-^{\text{s}} H^{\text{s}} N^{\text{s}}_+$, then by Bruhat decomposition of $G^s_{\mathbb{C}}$, $G^{\text{s}}_o$ is an open connected subset of $G^{\text{s}}$. Note that the map
\[\text{Ad: } G^{\text{s}}_o \to G^{\text{ad}}_o,\]
is a diffeomorphism.

\begin{lemma}[\cite{Kostant1979a}]
One has
\[(\text{Ad}^{-1} G^{\text{ad}}_o) \cap G^{\text{s}} = \mathcal{Z}(G^{\text{s}})G^{\text{s}}_o.\]
Furthermore if $c, c' \in \mathcal{Z}(G^{\text{s}})$ are distinct then $cG^{\text{s}}_o$ and $c'G^{\text{s}}_o$ are disjoint so that the connected components of $\mathcal{Z}(G^{\text{s}})G^{\text{s}}_o$ are uniquely of the form $cG^{\text{s}}_o$ for $c \in \mathcal{Z}(G^{\text{s}})$.
\end{lemma}

We can thus identify $G^{\text{s}}_o$ with $G^{\text{ad}}_o$ by the above diffeomorphism. This means that $N^{\text{s}}_-$ and $N^{\text{ad}}_-$, $H^{\text{s}}$ and $H^{\text{ad}}$ and also $N^{\text{s}}_+$ and $N^{\text{ad}}_-$ are all identified. We identify them in the following and drop all the superscripts whenever there is no confusion.

\subsubsection{Representation theory}\label{sec:representation}
The analytic weight lattice $\Gamma$ associated with $(U^{\text{s}}, T)$ consists of $\lambda \in \mathfrak{h}^*_{\mathbb{C}}$ satisfying the following two equivalent conditions (c.f. \cite{Knapp2002}):
\begin{enumerate}
\item $\lambda(X) \in 2 \pi i \mathbb{Z}$ whenever $X \in \mathfrak{t}$ satisfies $\exp X = 1$;
\item There is a multiplicative character $\xi_{\lambda}$ of $T$ with $\xi_{\lambda}(\exp X) = e^{\lambda(X)}$ for all $X \in \mathfrak{t}$.
\end{enumerate}
Extending $\lambda \in \mathfrak{h}^*_{\mathbb{C}}$ by linearity to $\mathfrak{h}_{\mathbb{C}}$, we obtain a linear functional $\lambda$ on $\mathfrak{h}_{\mathbb{C}}$  which takes real values on $\mathfrak{h}$ and the corresponding multiplicative character $\xi_{\lambda}$ on the subgroup $H^{\text{s}}_{\mathbb{C}}$ of $G^{\text{s}}_{\mathbb{C}}$ corresponding to $\mathfrak{h}_{\mathbb{C}}$. 

Recall that the algebraic lattice $P \subset \mathfrak{h}^*_{\mathbb{C}}$ is defined by
\[P = \left\{\lambda \in \mathfrak{h}^*_{\mathbb{C}} \left\vert\ \frac{2\kappa(\lambda, \alpha)}{\kappa(\alpha, \alpha)} \in \mathbb{Z}, \alpha \in \Delta\right.\right\}.\]
Now the fundamental weights $\omega_i \in P, 1 \le i \le \ell$, are defined by the relation $2\kappa(\omega_i, \alpha_j) \slash \kappa(\alpha_j, \alpha_j) = \delta_{ij}$ and we have the direct sum
\[P = \sum \limits_{i = 1}^{\ell}\mathbb{Z}\omega_i.\]
It is known that for the simply connected Lie group $G^{\text{s}}_{\mathbb{C}}$ we have $\Gamma = P$. The cone $P_+$ of dominant weights is defined as
\[P_+ = \{\lambda \in P \ \vert \ \kappa(\lambda, \alpha_i) \ge 0, \ 1 \le i \le \ell\}.\]
Then all the finite dimensional irreducible representations of $\mathfrak{g}$ (and also $G^s_{\mathbb{C}}$) are parameterized by $\lambda \in P_+$. Let $\mathbb{N}_0$ be the set of nonnegative integers, then we have
\[P_+ = \sum \limits_{i = 1}^{\ell} \mathbb{N}_0\omega_i.\]

Let $(\rho, V_{\mathbb{C}}^{\lambda})$ be a (finite dimensional) irreducible $\mathfrak{g}_{\mathbb{C}}$-module with highest weight $\lambda \in P_+$. Then a theorem due to Weyl known as the ``unitary trick'' states that the (finite-dimensional) representation of $G^s, U^s, \mathfrak{g}, \mathfrak{u}$, holomorphic representation of $G_{\mathbb{C}}^{\text{s}}$ and complex-linear representation of $\mathfrak{g}_{\mathbb{C}}$ on $V_{\mathbb{C}}^{\lambda}$ are all equivalent (c.f. \cite{Knapp2002}). A Hermitian (positive definite) inner product on $V^{\lambda}_{\mathbb{C}}$ which is invariant under $U^{\text{s}}$ can be introduced as follows: Let $\langle \cdot, \cdot\rangle$ be any Hermitian inner product on $V_{\mathbb{C}}^{\lambda}$ which is linear in the second factor and conjugate linear in the first one, then a Hermitian inner product satisfying the requirement can be obtained as
\begin{equation}\label{eq:Hermitian}
(u, v) = \int_{U^{\text{s}}} \langle \rho(x)\cdot u, \rho(x) \cdot v\rangle dx.
\end{equation}
Since elements in $U^{\text{s}}$ act as unitary operators on $V^{\lambda}_{\mathbb{C}}$ under the Hermitian inner product defined in \eqref{eq:Hermitian}, we can take an orthonormal basis for $V^{\lambda}_{\mathbb{C}}$ consisting of weight vectors of $\mathfrak{t}$. We fix once and for all a highest weight vector $v^{\lambda} \in V^{\lambda}_{\mathbb{C}}$ such that $(v^{\lambda}, v^{\lambda}) = 1$. 

Extending the Cartan involution $\theta$ on $\mathfrak{g}$ by linearity to an automorphism of $\mathfrak{g}_{\mathbb{C}}$, we then define a conjugate linear map on $\mathfrak{g}_{\mathbb{C}}$ by putting $X^* = \theta (-X^c)$ for any $X \in \mathfrak{g}_{\mathbb{C}}$. Then we have $X_{\alpha}^* = X_{-\alpha}$ and $H^*_{\alpha} = H_{\alpha}$. Since $G_{\mathbb{C}}^{\text{s}}$ is simply connected, the $*$-operation on $\mathfrak{g}_{\mathbb{C}}$ induces a unique $*$-operation on $G_{\mathbb{C}}^{\text{s}}$ such that $(\rho(g^*) \cdot v_1, v_2) = (v_1, \rho(g)\cdot v_2)$ for $g \in {G}_{\mathbb{C}}^{\text{s}}$.  For $g \in G$ and $v \in V^{\lambda}$, we write $\rho(g) \cdot v$ as $gv$ in the following for simplicity whenever there is no ambiguity.


\subsection{$\tau$-functions of f-KT hierarchy}\label{sec:tau}
Taking $g \in G_o$, then we can write
\[g = n_- h n_+\]
where $n_{\pm} \in N_{\pm}$ and $h \in H$. Note that $n_+ v^{\lambda} = v^{\lambda}$ and $(n_-)^* v^{\lambda} = v^{\lambda}$, we obtain (taking $h = \exp X$ for some $X \in \mathfrak{h}$)
\begin{equation}\label{eq:groupcharacter}
\xi_{\lambda}(h) = e^{\lambda(X)} = (v^{\lambda}, gv^{\lambda}).
\end{equation}

Let $e = \sum \limits_{i = 1}^{\ell}X_i$. For $L \in e + \mathfrak{b}_-$, let $\Theta_{L}(\t) = \sum \limits_{k = 1}^{\ell}t_{m_k}\nabla I_k(L) \in \mathfrak{g}$.
Let $L_{\0} := L_{\mathfrak{g}}({\0})$, and consider the following LU-factorization of the group element $\exp(\Theta_{L_{\0}}(\t))$ (which always exists for ${\t}$ small enough):
\begin{equation}\label{eq:LUfactorization}
\exp(\Theta_{L_{\0}}(\t)) = u({\t})b({\t}) \qquad \text{with } u({\t}) \in N_-, \ b({\t}) \in B_+.
\end{equation}

\begin{proposition}\label{prop:RHfactorization}
The solution $L_{\mathfrak{g}}({\t})$ of f-KT hierarchy is given by
\begin{equation}\label{eq:LUsolution}
L_{\mathfrak{g}}({\t}) = \text{Ad}_{u^{-1}({\t})} \cdot L_{\0} = \text{Ad}_{b({\t})} \cdot L_{\0}.
\end{equation}
\end{proposition}
\begin{proof}
Let $L({\t}) = \text{Ad}_{b({\t})} \cdot L_{\0}$, then $L({\t})$ obviously satisfies the same initial condition as $L_{\mathfrak{g}}({\t})$ and we would like to verify that it also satisfies the same differential equation as $L_{\mathfrak{g}}({\t})$. Taking derivative with respect to $t=t_1$ on both sides of the equation
\[\exp(\Theta_{L_{\0}}(\t))) = u({\t})b({\t}),\]
we obtain
\[\frac{d}{dt}\exp(\Theta_{L_{\0}}(\t)) = L_{\0} ub = ub L_{\0} = \dot{u}b + u\dot{b},\]
where $\dot{x}$ means the derivative of $x(\t)$ with respect to $t = t_1$. When $\t$ is small, this equation is meaningful for any simple Lie algebra and the associated connected Lie group: $L_{\0}\exp(\Theta_{L_{\0}}(\t))$ (resp. $\exp(\Theta_{L_{\0}}(\t))L_{\0}$) means we use $\exp(\Theta_{L_{\0}}(\t))$ to translate the vector $L_{\0}$ at the origin from the right (resp. left), and the equation says the effect is the same as the sum of the translation of two other vectors. We can rewrite the equation as
\[\text{Ad}_{u^{-1}({\t})} \cdot L_{\0} = \text{Ad}_{b({\t})} \cdot L_{\0} = u^{-1}\dot{u} + \dot{b}b^{-1}.\]
Note that these are all tangent vectors at the unit element of $G$ and can be viewed as elements in $\mathfrak{g} = \mathfrak{n}_- \oplus \mathfrak{b}_+$. By definition we obtain the decomposition of $L(\t)$ as
\[u^{-1}\dot{u} = \Pi_{\mathfrak{n}_-}L \quad \text{and} \quad \dot{b}b^{-1} = \Pi_{\mathfrak{b}_+}L.\]
Now differentiating $L = \text{Ad}_{b({\t})} \cdot L_{\0}$, we obtain
\[\frac{dL}{dt} = \dot{b}b^{-1}L - L \dot{b}b^{-1} = [\Pi_{\mathfrak{b}_+}L, L].\]
Thus $L(\t)$ satisfies the same differential equation as $L_{\mathfrak{g}}(\t)$, the uniqueness theorem of the differential equation implies that $L(\t) = L_{\mathfrak{g}}(\t)$.
\end{proof}

Now take $g$ in equation \eqref{eq:groupcharacter} to be $\exp(\Theta_{L_{\0}}(\t))$ and consider the decomposition $\exp(\Theta_{L_{\0}}(\t)) = n_-(\t)h(\t)n_+(\t)$ with $n_-(\t) = u(\t)$ and $h(\t)n_+(\t) = b(\t)$, we obtain
\[\xi_{\lambda}(h(\t)) = (v^{\lambda}, \exp(\Theta_{L_{\0}}(\t))v^{\lambda}).\]
Assume
\[h(\t) = \exp(\sum \limits_{i = 1}^{\ell}h_i(\t)H_i),\]
then
\begin{align*}
\xi_{\omega_i}(h(\t)) & = e^{\langle \omega_i, \sum \limits_{i = 1}^{\ell}h_i(\t) H_i\rangle}
 = e^{\langle \sum \limits_j C^{-1}_{ij}\alpha_j, \sum \limits_{k = 1}^{\ell}h_k(\t)H_k\rangle}\\
& = e^{\sum \limits_{j, k}h_k(\t) C^{-1}_{ij}C_{jk}}
 = e^{h_k(\t)\delta_{ik}}
 = e^{h_i(\t)}.
\end{align*}

For $\lambda = \omega_i, 1 \le i \le \ell$, taking diagonals of $\dot{b}b^{-1} = \Pi_{\mathfrak{b}_+}L$, we obtain $\dot{h}h^{-1} = \sum a_i(\t) H_i$, so that
\[\dot{h}_i(\t) = a_i(\t) \qquad \text{for } i = 1, \dots, \ell,\]
which gives $a_i(\t) = \frac{d}{dt}\ln \xi_{\omega_i}(h(\t))$. 
Thus
\begin{definition}
We define the $i$-th $\tau$-function of the f-KT hierarchy as
\begin{equation}\label{eq:tau0}
\tau_i({\t}) = (v^{\omega_i}, \exp(\Theta_{L_{\0}}(\t))v^{\omega_i}), \qquad 1 \le i \le \ell.
\end{equation}
\end{definition}

\begin{corollary}\label{cor:diagonal}
For $\t$ small enough, we have the following formula for the diagonal elements in f-KT hierarchy:
\begin{equation}\label{eq:ai}
a_i({\t}) = \frac{d}{dt}\ln \tau_i({\t}), \qquad 1 \le i \le \ell.
\end{equation}
\end{corollary}

\begin{remark}
Proposition \ref{prop:RHfactorization} and Corollary \ref{cor:diagonal} result from discussion with Y.Kodama, and we agree to use them both here and in \cite{Kodama-Okada2023} for quite different purpose.
\end{remark}

\begin{remark}
Note that in the above proof we did not use the specific form of the matrix $L(t)$ at all, which means the above formula is valid for much wider systems. Let $L(t) = \sum_{i = 1}^{\ell} a(t)H_i + \sum \limits_{\alpha \in \Delta} b_{\alpha}(t)X_{\alpha} \in \mathfrak{g}$ be any element, and define the Toda like equations
\[\frac{d L}{dt} = [B, L], \qquad \text{with} \qquad B = \Pi_{\mathfrak{b}_+}L,\]
where $\Pi_{\mathfrak{b}_+}$ is the projection to the upper Borel $\mathfrak{b}_+$ along $\mathfrak{n}_-$, then the diagonal elements are given by the following formula
\[a_i(t) = \frac{d}{dt}\ln \tau_i(t), \quad 1 \le i \le \ell, \qquad \text{for $t$ small enough.}\]
The f-KT lattice is special in that all the other matrix elements $b_{\alpha}$ of $L(t)$ are determined by the diagonal elements $a_i(t)$ (c.f. \cite{Xie2021} for an expression of $b_{\alpha}$ in terms of $\tau$-functions in type $A$). 
\end{remark}

\subsection{Singular solutions of the f-KT lattice}
Our main concern in the current paper is the structure of singular solutions of f-KT lattice which is equivalent to the divisor structure of $\tau$-functions by equation \eqref{eq:ai}.
To proceed we extend the base field to $\mathbb{C}$ and note the following theorem of Kostant:
\begin{proposition}[\cite{Kostant1978}]\label{thm:Kostantsection}
For a simple complex Lie algebra $\mathfrak{g}$ there exists an $\ell$-dimensional linear subspace $\mathfrak{s} \subset \mathfrak{b}_-$ such that elements in the affine subspace $e + \mathfrak{s}$ are regular. The map
\[\begin{array}{rcl}
 N_- \times (e + \mathfrak{s}) & \to &  e + \mathfrak{b}_-\\
 (n, X) & \mapsto & \text{Ad}_{n}X,
\end{array}\]
is an isomorphism of affine varieties.
\end{proposition}

\begin{remark}
When $\mathfrak{g} = \mathfrak{sl}_n(\mathbb{C})$, the choice of $\mathfrak{s}$ may be made so that $e + \mathfrak{s}$ is the space of traceless companion matrices. 
\end{remark}

Let $\mathcal{F}_{\Lambda}$ be the isospectral variety consisting Lax matrices with fixed Chevalley invariants $\Lambda = \{I_1, \dots, I_{\ell}\}$, i.e.
\[\mathcal{F}_{\Lambda} := \{L \in e + \mathfrak{b}_- \ \vert\ L \text{ has Chevalley invariants $\Lambda$}\}.\]

With any fixed choice of $\mathfrak{s}$, Proposition \ref{thm:Kostantsection} says that for any $L \in \mathcal{F}_{\Lambda}$ there exists a unique $u \in N_-$ and $C_{\Lambda} \in e + \mathfrak{s}$ such that 
\[L = u^{-1} C_{\Lambda}u,\]
and we can use this $\mathfrak{s}$ to embed the iso-spectral variety $\mathcal{F}_{\Lambda}$ into the flag variety $G \slash B_+$:
\begin{equation}\label{eq:clambda}
\begin{array}{rcl}
c_{\Lambda}: \mathcal{F}_{\Lambda} & \to & G \slash B_+\\
L & \mapsto & uB_+.
\end{array}
\end{equation}

We then have the following
\begin{proposition}
Each f-KT flow maps to the flag variety as

\[
\begin{tikzcd}
L_0 \ar[r, "{c_{\Lambda}}"] \ar[d, "{\text{Ad}_{u(\t)^{-1}}}"] & u_0 B_+ \ar[d]\\
L(\t) \ar[r, "{c_{\Lambda}}"] &\left\{ \begin{array}{l}
u_0u(\t)B_+\\
= u_0 \exp(\Theta_{L_0}(\t))B_+\\
= \exp(\Theta_{C_{\Lambda}}(\t))u_0B_+
\end{array}\right.
\end{tikzcd}
\]
where $L_0 = u_0^{-1}C_{\Lambda}u_0$.
\end{proposition}

When there exists some $\t_* \in \mathbb{C}^{\ell}$ such that the factorization \eqref{eq:LUfactorization} could not be performed, we need to use the following more general factorization
\[\exp(\Theta_{L_0}(\t_*)) = u_*\dot{w}_*b_*,\]
where $\text{id} \ne \dot{w}_* \in G$ is a representative of $w_* \in \mathcal{W}$. This means that the f-KT flow hits the boundary of a Bruhat cell which happens when there exists $1 \le k \le \ell$ such that $\tau_k(\t_*) = 0$ and the solution becomes singular at $\t = \t_*$. The set of $\t_*$ where some of the $\tau_k$ vanish is called the Painlev\'e divisor. Setting $\t \to \t + \t_*$, then
\[\tau_k(\t; w_*) = ( v^{\omega_k}, u_0^{-1} e^{\Theta_{C_{\Lambda}}(\t)} u_0 u_*\dot{w}_*b_* v^{\omega_k}).\]
Note that
\[b_* v^{\omega_k} = d_k v^{\omega_k}, \]
where $d_k \in \mathbb{C}$ is a constant, we obtain
\begin{equation}\label{eq:tau}
\tau_k(\t; w_*) = d_k(v^{\omega_k}, e^{\Theta_{C_{\Lambda}}(\t)}u\dot{w}_* v^{\omega_k}) \qquad 1 \le k \le \ell,
\end{equation}
where $u = u_0u_* \in N_-$. 

\begin{remark}
Note that a priori the proof of Proposition \ref{prop:RHfactorization} and thus the conclusion of Corollary \ref{cor:diagonal} can only be applied to regular solutions of f-KT hierarchy. To study singular solutions, we need to identify which $w \in \mathcal{W}$ and which $u \in N_-$ can appear in formula \eqref{eq:tau}. Counting the dimensions of $\mathcal{F}_{\Lambda}$ and $G \slash B_+$, it seems natural to expect that the compactification of $c_{\Lambda}(\mathcal{F}_{\Lambda})$ in $G \slash B_+$ is the whole $G \slash B_+$. However, this fact is somewhat nontrivial especially when the field is restricted to $\mathbb{R}$ as there are real regular solutions of f-KT hierarchy with dim$\, \mathfrak{b}_+$ many free parameters which do not hit any smaller Bruhat cells. We will use Kowalevski-Painlev\'e analysis to study singular solutions in Section \ref{sec:Kowalevski-Painleve} (c.f. \cite{Flaschka-Haine1991} for a discussion on the tridiagonal case).
\end{remark}


\section{Local study: Kowalevski-Painlev\'e analysis}\label{sec:Kowalevski-Painleve}
In this section we carry out a local study on the structure of Laurent series solutions of the f-KT lattice. The Kowalevski-Painlev\'e analysis for non-periodic tridiagonal Toda lattice was carried out by Flaschka and Zeng (c.f. \cite{Flaschka1988} \cite{Flaschka-Haine1991}), and the corresponding results for the periodic tridiagonal Toda lattice was obtained by Adler and van Moerbeke in the framework of algebraically integrable systems (c.f. \cite{Adler-vanMoerbeke1991}  \cite{Adler-vanMoerbeke-Vanhaecke2013}).

\subsection{Preliminaries on weight homogeneous systems}
Following \cite{Adler-vanMoerbeke-Vanhaecke2013}, we introduce some terminologies regarding weight homogeneous vector fields and their weight homogeneous Laurent solutions. 
\begin{definition}
\begin{enumerate}
\item Let $\nu = (\nu_1, \dots, \nu_n)$ be a collection of positive integers without a common factor. We say that a polynomial $f \in \mathbb{C}[x_1, \dots, x_n]$ is a weight homogeneous polynomial of weight $k$ (with respect to $\nu$) if
\[f(\alpha^{\nu_1}x_1, \dots, \alpha^{\nu_n}x_n) = \alpha^kf(x_1, \dots, x_n),\]
for all $(x_1, \dots, x_n) \in \mathbb{C}^n$ and $\alpha \in \mathbb{C}$. We denote the weight of $f$ by $\varpi(f)$ henceforth with $\nu$ and $n$ being fixed. 

\item A polynomial vector field $\mathcal{V}$ in $\mathbb{C}^n$, given by
\begin{align}\label{eq:weighthom}
\frac{d}{dt}x_i = f_i(x_1, \dots, x_n), \qquad (i = 1, \dots, n)
\end{align}
is called a weight homogeneous vector field of weight $k$ (with respect to $\nu$) if each of the polynomials $f_1, \dots, f_n$ is weight homogeneous (with respect to $\nu$) and if $\varpi(f_i) = \nu_i + k = \varpi(x_i) + k$ for $i = 1, \dots, n$. 
 
\item If \eqref{eq:weighthom} is a polynomial weight homogeneous vector field of weight $1$, then a Laurent series solution to \eqref{eq:weighthom} of the form
\begin{align}\label{eq:weightLaurent}
x_i(t) = \frac{1}{t^{\nu_i}} \sum \limits_{k = 0}^{\infty}x_{ik}t^k, \qquad i = 1, \dots, n
\end{align}
with $x_0 = (x_{10}, \dots, x_{n0}) \ne 0$, is called a weight homogeneous Laurent solution.
\end{enumerate}
\end{definition}

In the following we call a polynomial weight homogeneous vector field of weight $1$ a weight homogeneous vector field for brevity. 
We list several known properties for a weight homogeneous vector field and their weight homogeneous Laurent solutions.

\begin{proposition}[\cite{Adler-vanMoerbeke-Vanhaecke2013}, Proposition 7.6]\label{prop:Laurentsolutions}
Suppose that $\mathcal{V}$ is a weight homogeneous vector field on $\mathbb{C}^n$, given by \eqref{eq:weighthom}, and $x(t) = (x_1(t), \dots, x_n(t))$ is a weight homogeneous Laurent solution for $\mathcal{V}$. Then the leading coefficients $x_{i0}$ satisfy the non-linear algebraic equations
\begin{align}\label{eq:Indicial}
& \nu_1x_{10} + f_1(x_{10}, \dots, x_{n0}) = 0,\nonumber\\
& \qquad \qquad \vdots \\
&  \nu_nx_{n0} + f_n(x_{10}, \dots, x_{n0}) = 0, \nonumber
\end{align}
while the subsequent terms $x_{ik}$ satisfy
\begin{align}\label{eq:Higherorderterm}
\left(k \cdot I_n - \mathcal{K}(x_{(0)})\right)x_{(k)} = \vec{R}_{(k)},
\end{align}
where $x_{(k)} = \begin{pmatrix}x_{1k}\\ \vdots \\ x_{nk} \end{pmatrix}$ and $\vec{R}_{(k)} = \begin{pmatrix} R_{1k}\\ \vdots \\ R_{nk}\end{pmatrix}$ are $n \times 1$ column vectors and $I_n$ is the $n \times n$ identity matrix; each $R_{ik}$ is a polynomial, which depends on the variables $x_{1l}, \dots, x_{nl}$ with $0 \le l < k$ only. Also, the $(i, j)$-th entry of the $n \times n$-matrix $\mathcal{K}$ is the regular function on $\mathbb{C}^n$, defined by
\begin{align}\label{eq:Kowalevski-matrix}
\mathcal{K}_{ij} := \frac{\partial f_i}{\partial x_j} + \nu_i\delta_{ij},
\end{align}
where $\delta$ is the Kronecker delta.
\end{proposition}
The set of equations \eqref{eq:Indicial} is called the indicial equation of $\mathcal{V}$. Its solution set is called the indicial locus, and is denoted by $\mathcal{I}$. The $n \times n$ matrix $\mathcal{K}$ defined in \eqref{eq:Kowalevski-matrix} is called the Kowalevski matrix.

\begin{theorem}[\cite{Adler-vanMoerbeke-Vanhaecke2013}]\label{thm:K-matrixgeneral}
In the above setting, the following results regarding the spectrum of $\mathcal{K}$ and the weight homogeneous Laurent solutions are known.
\begin{enumerate}
\item

For any $m \in \mathcal{I}$, the Kowalevski matrix $\mathcal{K}(m)$ always has $-1$ as an eigenvalue. The corresponding eigenspace contains $(\nu_1m_1, \dots, \nu_nm_n)^{T}$ as an eigenvector (\cite{Adler-vanMoerbeke-Vanhaecke2013}, Proposition 7.11).

\item If we have $p$ constants of motion of weight $k$ whose differentials are independent at $m \in \mathcal{I}$, then $k$ is an eigenvalue of $\mathcal{K}(m)$ with multiplicity at least $p$ (\cite{Adler-vanMoerbeke-Vanhaecke2013}, Theorem 7.30).

\item Let $m \in \mathcal{I}$ be an arbitrary element, then
\begin{align}\label{eq:leadingterm}
m(t) := \left(\frac{m_1}{t^{\nu_1}}, \cdots, \frac{m_n}{t^{\nu_n}}\right),
\end{align}
is a solution to \eqref{eq:weighthom} for $t \ne 0$ (\cite{Adler-vanMoerbeke-Vanhaecke2013}, Proposition 7.14).

\item
For a polynomial weight homogeneous vector field \eqref{eq:weighthom}, their weight homogeneous Laurent solutions $x(t)$ given by \eqref{eq:weightLaurent} converge (\cite{Adler-vanMoerbeke-Vanhaecke2013}, Theorem 7.25).

\end{enumerate}
\end{theorem}

\begin{remark}
Note that in \cite{Adler-vanMoerbeke-Vanhaecke2013} the authors only consider singular solutions for what they called algebraically integrable systems, and for holomorphic solutions (which is an important part of all Laurent series solutions), $-1$ is of course not part of the spectrum of $\mathcal{K}$. Note also that f-KT lattice as we consider here is not an algebraically integrable system in the sense of Adler and van Moerbeke (c.f. \cite{Adler-vanMoerbeke-Vanhaecke2013} for definition), and in some sense it is a degenerate one (c.f. \cite{Ercolani-Flaschka-Singer1993} \cite{Reshetikhin2016} for an explanation).
\end{remark}

\subsection{f-KT lattice as a weight homogeneous system}
Now we come back to our study on the f-KT lattice.
Note that the f-KT lattice \eqref{eq:f-KThky}
can be written as
\begin{align}\label{eq:f-KTincoordinate}
& \frac{d a_i}{d t} = b_i \quad (1 \le i \le \ell), \nonumber\\
& \frac{d b_{\alpha}}{d t} = - \sum_{i = 1}^{\ell}\left(\alpha(H_{\alpha_i})a_i\right)b_{\alpha} + \sum \limits_{i = 1}^{\ell}N_{\alpha_i, -\alpha - \alpha_i}b_{\alpha + \alpha_i}, \quad \alpha \in \Delta_+,
\end{align}
where $b_i = b_{\alpha_i}, [H_{\alpha_i}, Y_{\alpha}] = -\alpha(H_{\alpha_i})Y_{\alpha}$ and $[X_{\alpha_i}, Y_{\alpha + \alpha_i}] = N_{\alpha_i, -\alpha - \alpha_i}Y_{\alpha}$ for $\alpha, \alpha + \alpha_i \in \Delta_+$. In the Chevalley basis chosen in Section \ref{sec:f-KT}, for independent roots $\alpha$ and $\beta$, if $\beta - q\alpha, \dots, \beta + p\alpha$ is the $\alpha$-string through $\beta$, then $N_{\alpha, \beta} = \pm (q+1)$ if $p \ge 1$ and $N_{\alpha, \beta} = 0$ if $\alpha + \beta \not\in \Delta$.

Recall that the height of $\alpha \in \Delta$ is defined as $L(\alpha) = \sum_{i=1}^{\ell}c_i$ for $\alpha=\sum_{i=1}^{\ell}c_i\alpha_i$. 
Equations (\ref{eq:f-KTincoordinate}) by definition define a polynomial weight homogeneous vector field of weight $1$ with $\varpi(a_{i}) = 1$ and $\varpi(b_{\alpha}) = L(\alpha) + 1 = \nu_{\alpha}$, and a weight homogeneous Laurent series solution has the following form
\begin{align}\label{eq:Laurentansatz}
a_{i}(t) = \frac{1}{t}\sum \limits_{k = 0}^{\infty} a_{ik}t^k, \qquad b_{\alpha}(t) = \frac{1}{t^{\nu_{\alpha}}} \sum \limits_{k = 0}^{\infty}b_{\alpha k}t^k.
\end{align}

The following proposition shows that the weight homogeneous Laurent solutions \eqref{eq:Laurentansatz} are the only Laurent series solutions the f-KT lattice \eqref{eq:f-KTincoordinate} may have.
\begin{proposition}\label{prop:Uniqueness}
Solutions of the form (\ref{eq:Laurentansatz}) exist, and all Laurent series solutions of f-KT are of this form (Note that we don't require ${ a_{i0} \ne 0}$ or ${ b_{\alpha 0} \ne 0}$ here).
\end{proposition}
\begin{proof}
The existence part will be dealt with later. Note from the differential equations \eqref{eq:f-KTincoordinate} that $b_{\alpha}$'s can be recursively solved once $a_{i}$'s are known. So to prove this Proposition, we essentially only need to show that the worst singularities for $a_{i}$'s are just simple poles (singularity of the first order). We prove this by contradiction. To detect the worst singularities, we may assume at least one of the $a_i$'s has a pole of order bigger than or equal to $2$ at $t = 0$, i.e. 
\begin{align}\label{eq:aalphai}
a_{i}(t) = \frac{1}{t^{\delta}}\sum \limits_{k = 0}^{\infty} a_{ik}t^k,
\end{align}
where $\delta \ge 2$ for $1 \le i \le \ell$ and there exists $1 \le i = \kappa \le \ell$ such that $a_{\kappa 0} \ne 0$. 

The equations for $b_{\alpha_j}$ then read as
\begin{align} \label{eq:balpha}
\frac{d b_{\alpha_j}}{d t} = - \sum_{i = 1}^{\ell}\left(a_i\alpha_j(H_{\alpha_i})\right)b_{\alpha_j} + \sum \limits_{i = 1}^{\ell}N_{\alpha_i, -\alpha_j - \alpha_i}b_{\alpha_j + \alpha_i}.
\end{align}
Substituting the Laurent series \eqref{eq:aalphai} into \eqref{eq:balpha}, we observe that on the left hand side of the $\kappa$-th equation of \eqref{eq:balpha} it has a pole of order at most $\delta_{\kappa} + 2$ with respect to $t$ at $t = 0$, while on the right hand side the term $a_{\kappa}b_{\alpha_{\kappa}}$ produces a pole of order $2\delta_{\kappa} + 1 \ge \delta_{\kappa} + 3 > \delta_{\kappa} + 2$ by assumption, so we need to have another term from the right hand side to balance this higher order singular term $a_{\kappa}b_{\alpha_{\kappa}}$ by requiring that either there exists another $a_{j0} \ne 0 \ (j \ne \kappa)$ or one of the terms $b_{\alpha_j + \alpha_{\kappa}}$ could do the job.
In either case, this will cause another equation in \eqref{eq:balpha} to have higher order singularities on the right hand side than the left hand side. This argument continues, and in the end we see that all equations in \eqref{eq:balpha} have terms with higher order poles on the right hand side than the left hand side. Comparing the coefficients of $t^{-2\delta - 1}$, we get
\begin{align*}
-\delta \sum_{i = 1}^{\ell}\left(\alpha_j(H_{\alpha_i})a_{i0}\right)a_{j0} + \sum \limits_{i = 1}^{\ell}N_{\alpha_i, -\alpha_j - \alpha_i}b_{\alpha_j + \alpha_i,0} = 0, \qquad 1 \le j \le \ell.
\end{align*}
Since there are at most $\ell -1$ non-vanishing $b_{\alpha_j + \alpha_i}$, we can eliminate them from the above system to obtain a quadratic homogeneous equation satisfied by $a_{i0}$'s.
Now we continue this process to deal with equations satisfied by $b_{\alpha}$ with root length $L(\alpha) = 2, 3, \dots$ until we reach the unique highest root $\alpha_0$.  When the number of roots in height $k+1$ equals the number of roots in height $k$, we can solve $b_{\alpha}$ ($L(\alpha) = k+1$) uniquely and substitute them into equations on the next height. And we get at least $1$ homogeneous equation for $a_{i0}$'s when the number of roots in height $k+1$ is less than the number of roots in height $k$ (We get two homogeneous equations of the same degree in type $D$ at places where the number of roots jumps by $2$ at contiguous heights). Eventually, we get $\ell$ homogeneous polynomial equations in the $\ell$ variables $(a_{10}, a_{20}, \dots, a_{\ell0})$ after we eliminate all the $b_{\alpha 0}$'s. It can then be checked that $a_{i0} = 0, 1 \le i \le \ell$ is the only solution for the $\ell$ homogeneous polynomial equations when $\delta > 1$ which leads to the desired contradiction.
\end{proof}

\begin{remark}
After finishing this paper we noted that the fact that $a_i(t)$'s have at worst simple pole singularities had already been observed by Singer in \cite{Singer1991} with similar argument and we keep the proof here to be consistent with the style in the following sections.  
\end{remark}

\subsection{The indicial equations for f-KT lattice}

Now we can substitute the Laurent series \eqref{eq:Laurentansatz} into the differential equations \eqref{eq:f-KTincoordinate} and follow Proposition \ref{prop:Laurentsolutions} to get Laurent solutions of \eqref{eq:f-KTincoordinate}. As a first step we need to deal with the indicial equations \eqref{eq:Indicial} for f-KT lattice, more explicitly we have
\begin{equation}\label{eq:Indicialf-KT}
\begin{array}{ll}
 a_{i0} + b_{i0} = 0 \qquad & 1 \le i \le \ell, \\
 \nu(\alpha)b_{\alpha 0} - \sum_{i = 1}^{\ell}\left(\alpha(H_{\alpha_i})a_{i0}\right)b_{\alpha 0} + \sum \limits_{i = 1}^{\ell}N_{\alpha_i, -\alpha - \alpha_i}b_{\alpha + \alpha_i, 0} = 0,\qquad & \alpha + \alpha_i \in \Delta_+.
\end{array}
\end{equation}
Note that the values of $b_{\alpha 0}$'s are uniquely determined by values of $a_{i0}$, so the vector $\vec{a}_0 = (a_{10}, \dots, a_{\ell 0})$ will be our concern in the following. 

We first note the following interesting fact regarding the cardinality of the Weyl group $\mathcal{W}$ of a simple Lie algebra $\mathfrak{g}$.
\begin{lemma}\label{lem:Weyl}
List all the positive roots of a simple Lie algebra $\mathfrak{g}$ by their heights, and assume that there are $m$ different heights, and there are $\ell_i$ positive roots in height $i$ $1 \le i \le m$, then
\begin{align*}
\vert \mathcal{W} \vert = 1^{\ell - \ell_1}2^{\ell_1 - \ell_2}3^{\ell_2 - \ell_3} \cdots m^{\ell_{m-1} - \ell_m} (m+1)^{\ell_m}.
\end{align*}
\end{lemma}
\begin{remark}
The fact in Lemma \ref{lem:Weyl} is surely known to experts but we couldn't identify a proper reference, and it is not hard to check it for all simple Lie algebras directly. For example, in Lie algebra of type $A_{\ell}$, we have $m = \ell$, $\ell_i = \ell+1-i$ and $\vert \mathcal{W} \vert = (\ell+1)!$.
\end{remark}

Now we have
\begin{proposition}\label{conj:Nonnegative}
The indicial equations \eqref{eq:Indicialf-KT} have ${ \vert \mathcal{W} \vert}$ many solutions where $\mathcal{W}$ is the Weyl group of $\mathfrak{g}$, and all of them are non-negative integral-valued vectors. 
\end{proposition}
\begin{proof}
Following the proof of Proposition \ref{prop:Uniqueness}, we can eliminate all the $b_{\alpha 0}$'s from the indicial equations \eqref{eq:Indicialf-KT} and obtain $\ell$ polynomial equations for $a_{i0}$'s (The difference from Proposition \ref{prop:Uniqueness} is that when $\delta = 1$ the $\ell$ equations are not homogeneous anymore, thus admitting nontrivial solutions). Note that we get $k$ degree $j$ polynomial equation for $a_{i0}$'s exactly when $\ell_{j-1} - \ell_{j} = k$. By Lemma \ref{lem:Weyl} and B\'ezout's theorem, we know that the number of solutions of the indicial equations is at most $\vert \mathcal{W} \vert$. 

To show the second part of this Proposition, we construct $\vert \mathcal{W} \vert$ many different solutions for \eqref{eq:Indicialf-KT} explicitly. 

Let
\begin{align}\label{eq:tildetau}
\tilde\tau_i(t; w) = (v^{\omega_i}, \exp(te)\dot{w} \cdot v^{\omega_i}) = t^{a^w_{i0}},
\end{align}
where $e = \sum \limits_{i = 1}^{\ell}X_{\alpha_i}$, $\dot{w} \in G$ is a representative of $w \in \mathcal{W}$, $v^{\omega_i}$ is the highest weight vector in the $i$-th fundamental representation $V^{\omega_i}$ and $(\cdot, \cdot)$ is the Hermitian product on $V^{\omega_i}$ introduced in Section \ref{sec:representation}. Then it can be checked directly that \begin{align*}
a^w_i(t) = \frac{d}{dt}\ln\tilde{\tau}_i(t; w) = \frac{a^{w}_{i0}}{t}
\end{align*}
is a solution to \eqref{eq:f-KTincoordinate} in the form of Theorem \ref{thm:K-matrixgeneral} and $a^w_{i0}$ is a solution to \eqref{eq:Indicialf-KT}.
\end{proof}

\begin{remark}
Proposition \ref{conj:Nonnegative} shows that for $\Lambda \in \mathbb{C}^{\ell}$ the isospectral variety $\mathcal{F}_{\Lambda}$ intersects with all the Painlev\'e divisors, equivalently this means that $\overline{c_{\Lambda}(\mathcal{F}_{\Lambda})}$ intersects with all the Bruhat cells.
\end{remark}

Note that Proposition \ref{conj:Nonnegative} also means that each solution $\vec{a}_0 = (a_{10}, \dots, a_{\ell 0})$ of the algebraic system \eqref{eq:Indicialf-KT} is associated with a Weyl group element and is a set of non-negative integers. This is a non-trivial fact and we would like to explore its combinatoric meaning in the following.

We denote by $\Phi^{\pm}_w$ the set of positive (negative) coroots of $(\mathfrak{g}, \mathfrak{h})$ which are mapped into the set of negative (positive) coroots of $(\mathfrak{g}, \mathfrak{h})$ by $w$, then $\Phi^{\pm}_w$ and their complements are closed under addition in the set of positive (negative) coroots by which we mean the sum of any two elements in each set either belongs to the same set or is not in the coroot system. 

The more precise one-to-one correspondence between Weyl group elements and solutions of indicial equations \eqref{eq:Indicialf-KT} comes in the following manner.
\begin{theorem}\label{thm:Onetoone}
Let 
\begin{align*}
\sum \limits_{\check{\alpha} \in \Phi^+_w}\check{\alpha} = \sum \limits_{j = 1}^{\ell} \check{q}_j\check{\alpha}_j, \quad \check{q}_j \in \mathbb{N}_0,
\end{align*}
where ${ \check{\alpha}}$ denotes the coroot associated with ${ \alpha \in \Delta_+}$.
Then the vectors ${ \vec{q}_w = (\check{q}_1, \dots, \check{q}_{\ell}), w \in \mathcal{W}}$ are solutions $\vec{a}_0 = (a_{10}, \dots, a_{\ell 0})$ to the indicial equations \eqref{eq:Indicialf-KT} and they give all solutions to the indicial equations. 

We denote the solution of indicial equations associated with $w \in \mathcal{W}$ by $\vec{a}^w_{0}$, then reformulating the above result we have
\begin{align}\label{eq:ai0w}
a_{i0}^w = \sum \limits_{\check{\alpha} \in \Phi_w^-}\kappa(\omega_i, w^{-1}\check{\alpha}),
\end{align}
where $\omega_i$ is the $i$th fundamental weight, and $\kappa(\cdot, \cdot)$ is the bilinear form on $\mathfrak{h}^*$ induced by the restriction of the Killing form to $\mathfrak{h}$ which is invariant under the Weyl group action. 
\end{theorem}
\begin{proof}
We would like to identify the degrees of the $\tilde{\tau}_i(t; w)$'s introduced in the proof of Proposition \ref{conj:Nonnegative} explicitly. 
For every $\alpha_i \in \Pi$ there is an associated homomorphism
\[\varphi_i: {SL}_2 \to {G}.\]
Consider the $1$-parameter subgroups in ${G}$ defined by
\begin{align*}
x_i(t) = \varphi_i\begin{pmatrix}1 & t \\ 0 & 1 \end{pmatrix}, \quad y_i(t) = \varphi_i\begin{pmatrix} 1 & 0 \\ t & 1\end{pmatrix}, \quad \check{\alpha}_i(t) = \varphi_i\begin{pmatrix} t & 0 \\
0 & t^{-1}\end{pmatrix}.
\end{align*}
Let $s_i \in \mathcal{W}$ be a simple reflection, then its representative in ${G}$ is given by
\[\dot{s}_i := \varphi_i\begin{pmatrix} 0 & 1 \\ -1 & 0 \end{pmatrix}.\]
Then we have
\begin{align}\label{eq:anidentity}
\dot{s}_i\check{\alpha}_i(t^{-1}) = y_i(-t^{-1})x_i(t)y_i(-t^{-1}).
\end{align}

For $w \in \mathcal{W}$, we choose a reduced expression $w = s_{i_1}s_{i_2} \cdots s_{i_n} \in \mathcal{W}$ and set $\dot{w} = \dot{s}_{i_1}\dot{s}_{i_2} \cdot \dots \cdot \dot{s}_{i_n}$ to be a representative of $w$ in ${G}$. It is well-known that this product is independent of the choice of reduced expressions for $w$ (c.f. \cite{Lusztig1994}).

We would like to show that
\begin{align*}
\tilde\tau_j(t; w) \varpropto \prod \limits_{l = 1}^n t^{\kappa(s_{i_{l+1}}s_{i_{l+2}} \cdots s_{i_n}\omega_j, \check{\alpha}_{i_l})}.
\end{align*}
where $\tilde{\tau}_j$ is defined in \eqref{eq:tildetau} and $\check{\alpha}_{i_l}$ is the $i_l$-th coroot.

First we note that
\begin{align*}
\tilde\tau_j(t; w) \varpropto \sigma_j(t; w) := ( v^{\omega_j}, x_{i_n}(t)x_{i_{n-1}}(t) \cdots x_{i_1}(t)\dot{s}_{i_1}\dot{s}_{i_2} \cdots \dot{s}_{i_n} v^{\omega_j}).
\end{align*}
So our goal in the following is to show that
\begin{align}\label{eq:sigmaJ}
\sigma_j(t; w) = \prod \limits_{l = 1}^n t^{\kappa(s_{i_{l+1}}s_{i_{l+2}} \cdots s_{i_n}\omega_j, \check{\alpha}_{i_l})}.
\end{align}

For $k = 1, 2, \dots, n+1$, we prove by reverse induction on $k$, that
\begin{align}\label{eq:sigmaK}
\sigma_j(t; w^{(k)}) = \prod \limits_{l = k}^n t^{\kappa(s_{i_{l+1}}s_{i_{l+2}} \cdots s_{i_n}\omega_j, \check{\alpha}_{i_l})}.
\end{align}
where ${w}^{(k)} = {s}_{i_k}{s}_{i_{k+1}} \cdots {s}_{i_n}$. The start of the induction is clear. Suppose that the result holds for $k+1$, i.e., for $w^{(k+1)} = \dot{s}_{i_{k+1}} \cdots \dot{s}_{i_n}$, and consider $\dot{w}^{(k)} = \dot{s}_{i_k}\dot{s}_{i_{k+1}} \cdots \dot{s}_{i_n}$. Denote $g^{(k)} = x_{i_n}(t)x_{i_{n-1}}(t) \cdots x_{i_k}(t)$. Then, since $w^{(k+1)}\omega_j$ and $w^{(k)}\omega_j$ are extremal weights, and $w^{(k)}\omega_j = s_{i_k}w^{(k+1)}\omega_j \le w^{(k+1)}\omega_j$ we have
\begin{align*}
\sigma_j(t; w^{(k)}) & = (v^{\omega_j}, g^{(k+1)}x_{i_k}(t)\dot{w}^{(k)} \cdot v^{\omega_j})\\
& = ( v^{\omega_j}, g^{(k+1)}\dot{w}^{(k+1)} \cdot v^{\omega_j}) (\dot{w}^{(k+1)} \cdot v^{\omega_j}, x_{i_k}(t)\dot{w}^{(k)} \cdot v^{\omega_j} ).
\end{align*}

By equation \eqref{eq:anidentity}, we have
\begin{align*}
\sigma_j(t; w^{(k)}) & = (v^{\omega_j}, g^{(k+1)}\dot{w}^{(k+1)} \cdot v^{\omega_j}) (\dot{w}^{(k+1)} \cdot v^{\omega_j}, y_{i_k}(t^{-1})\dot{s}_{i_k}\check{\alpha}_{i_k}(t^{-1})y_{i_k}(t^{-1})\dot{w}^{(k)} \cdot v^{\omega_j} )\\
& = t^{\kappa(s_{i_{k+1}} \dots s_{i_n}\omega_j, \check{\alpha}_{i_k})}\sigma_j(t; {w}^{(k+1)}).
\end{align*}
Equation \eqref{eq:sigmaK} for $k$ now follows from the inductive hypothesis.

From equation \eqref{eq:sigmaJ}, we get
\begin{align*}
\sigma_j(t; {w}) = \prod \limits_{l = 1}^n t^{\kappa(\omega_j, s_{i_m} \cdots s_{i_{l+2}}s_{i_{l+1}}\check{\alpha}_{i_l})}.
\end{align*}
Now we note that
\begin{align*}
\sum \limits_{l = 1}^n s_{i_{n}} \cdots s_{i_{l+2}}s_{i_{l+1}}\check{\alpha}_{i_l}
\end{align*}
records in a step by step manner which positive coroots are mapped to the negative ones by $w$, more explicitly
\begin{align*}
w \cdot s_{i_{n}} \cdots s_{i_{l+2}}s_{i_{l+1}}\check{\alpha}_{i_l} & = s_{i_1} \cdots s_{i_n} \cdot s_{i_{n}} \cdots s_{i_{l+2}}s_{i_{l+1}}\check{\alpha}_{i_l}\\
& = s_{i_1} \cdots s_{i_l}\check{\alpha}_{i_l},
\end{align*}
which is a negative coroot for any $1 \le l \le n$ and they are all of them. We are done!
\end{proof}
\begin{remark}
The essential ingredient of the proof of Theorem \ref{thm:Onetoone} is contained in the proof of Lemma 7.5 \cite{Marsh-Rietsch2004}, and with proper notion it turns out to be a very special case of a general formula for generalized minors there. 
\end{remark}

\begin{remark}
In Theorem \ref{thm:Onetoone}, we relate the solutions of the indicial equations with the actions of the Weyl group elements on the coroot system. Since at the Lie algebra level, the coroot system of $\mathfrak{g}$ is isomorphic to the root system of its Langlands dual $\check{\mathfrak{g}}$ whose Cartan matrix is the transpose of the Cartan matrix of $\mathfrak{g}$, the relation we found here might be interesting in other branches of mathematics. Similar links between tridiagonal Toda lattice and coroots were observed for example in \cite{Flaschka-Haine1991} and \cite{Casian-Kodama2006}.
\end{remark}

Y. Li pointed out to me that properties of the sets $\Psi^{\pm}_w$ (the dual of $\Phi^{\pm}_w$) were investigated in \cite{Papi1994} by P. Papi. 
Let $\Psi^{\pm}_w$ be the set of positive (negative) roots of $(\mathfrak{g}, \mathfrak{h})$ which are mapped into the set of negative (positive) ones by $w^{-1}$. Let
 \[\sum \limits_{\alpha \in \Psi^+_w}\alpha = \sum \limits_{j = 1}^{\ell}q_j\alpha_j, \]
then Papi showed that the sets $\Psi^{\pm}_w \subseteq \Delta_{\pm}$ associated with $w \in \mathcal{W}$ are uniquely determined by the property that these sets and their complements in $\Delta_{\pm}$ respectively are closed under addition. More precisely,
\begin{theorem}[\cite{Papi1994}]
$S \subseteq \Delta_+$ is associated with some $w \in \mathcal{W}$ if and only if it satisfies the following properties:
\begin{enumerate}
\item If $\alpha, \beta \in S$ and $\alpha + \beta \in \Delta$, then $\alpha + \beta \in S$.
\item If $\alpha + \beta \in S, \alpha, \beta \in \Delta_+$, then $\alpha \in S$ or $\beta \in S$.
\end{enumerate}
Moreover, any such set is associated to a unique $w \in W$.
\end{theorem}

A consequence of Theorem \ref{thm:Onetoone} is that the vector $\vec{q}_w = (\check{q}_1, \dots, \check{q}_{\ell})$ uniquely determines $w \in \mathcal{W}$ which can also be proved directly.

Papi then showed that there exists a one-to-one correspondence between reduced expressions of elements $w \in \mathcal{W}$ and orderings of the set $\Psi^+_w$ satisfying
\begin{enumerate}
\item If $\alpha, \beta \in \Psi^+_w,\ \alpha < \beta$, and $\alpha + \beta \in \Delta$, then $\alpha + \beta \in \Psi^+_w$ and $\alpha < \alpha + \beta < \beta$.
\item If $\alpha + \beta \in \Psi^+_w,\ \alpha, \beta \in \Delta_+$, then $\alpha$ or $\beta$ (or both) belong to $\Psi^+_w$ and one of them precedes $\alpha + \beta$.
\end{enumerate}
More precisely, let $s_{i_1} \cdots s_{i_m}$ be a reduced expression of $w \in \mathcal{W}$, then the order on $\Psi^+_w$ associated with this reduced expression is
\[\beta_1 := \alpha_{i_1}; \qquad \beta_k := s_{i_1} \cdots s_{i_{k-1}}(\alpha_{i_k}), \quad k = 2, \dots, m.\]

\begin{remark}
It is interesting to see the consequence of this refined correspondence between $\mathcal{W}$ and ordered subset of $\Delta_+$ on f-KT lattice. This might be linked with the Deodhar decomposition of the real flag variety and asymptotic properties of the real solutions of the f-KT lattice (c.f. \cite{Kodama-Williams2015} for type $A$ case).
\end{remark}


\subsection{Spectra of the Kowalevski matrix}
Up to now we have given rather complete characterizations for the leading coefficients in the Laurent series solutions, we would like to understand more about higher order terms of the homogeneous Laurent solutions by exploring the spectrum of the Kowalevski matrices $\mathcal{K}_w,\ w \in \mathcal{W}$.
For f-KT lattice we first make the following definition for the spectrum of Kowalevski matrices.
\begin{lemma}[Definition]
The eigenvalues of ${ \mathcal{K}_w}$ are all integers and they can be divided into two types: ${\ell}$ of them are constants given by the degrees of the Chevalley invariants, and we say that they belong to type ``C"; the others depend on solutions of the indicial equations which are a priori unknown, and we call them type ``X" eigenvalues.
\end{lemma}

Our next goal is to understand these type ``X'' eigenvalues. The following characterization of the type-X eigenvalues is originally obtained by some brutal force calculations.
\begin{theorem}[\cite{Xie2021}]\label{conj:Kmatrix}
There are dim ${ \mathfrak{n}_-}$ many type-X eigenvalues and each of them is associated with a positive (resp. negative) root $\alpha \in \Delta_+$ of $(\mathfrak{g}, \mathfrak{h})$. Denote by $E^{w}_{\alpha}$ the eigenvalue of $\mathcal{K}_w$ associated with $\alpha \in \Delta_+$, then $E^w_{\alpha}$ is given by 
\begin{equation}\label{eq:Xexponent1}
E^w_{\alpha} = -{L(\alpha) - \sum_{i = 1}^{\ell}\alpha(h_{\alpha_i})a^w_{i0}} \quad \text{with} \quad { \alpha \in \Delta_+}.
\end{equation}
\end{theorem}

An equivalent form of \eqref{eq:Xexponent1} is given as follows.
\begin{theorem}\label{thm:Kmatrix}
For each $w \in \mathcal{W}$ and $\alpha \in \Delta_+$ we have
\begin{equation}\label{eq:Xexponent2}
E^w_{\alpha} = L(w\alpha).
\end{equation}

\end{theorem}

We make several remarks before giving the proof of Theorem \ref{thm:Kmatrix}.

\begin{remark}\label{rem:equivariant}
An immediate consequence of formula \eqref{eq:Xexponent2} is that there is a Weyl group action acting on the ``X'' type eigenvalues $E^w_{\alpha}$ which is equivariant to the Weyl group action on the root system of $(\mathfrak{g}, \mathfrak{h})$. 
\end{remark}

\begin{corollary}\label{cor:freeparameters}
Let ${ \Lambda_w}$ be the set of positive eigenvalues of ${ \mathcal{K}_w}$, then for the Laurent series solutions corresponding to ${ w \in \mathcal{W}}$ we have ${ \vert \Lambda_w \vert = \text{dim } \mathfrak{b}_- - l(w)}$ where $l(w)$ is the length of the Weyl group element. Note that since the number of positive (negative) roots is the same as the length of the longest Weyl group element $w_0$, we have ${ \vert \Lambda_w \vert = \ell + l(w_0) - l(w)}$ where $\ell$ comes from the $\ell$ type ``C'' parameters.
\end{corollary}

Now we give a proof for Theorem \ref{thm:Kmatrix}. The link between Theorem \ref{conj:Kmatrix} and Theorem \ref{thm:Kmatrix} is built through the following lemma.
\begin{lemma}\label{lem:weylroots}
For any $w \in \mathcal{W}$, we have
\[(w - 1) \sum \limits_{i = 1}^{\ell}\omega_i = \sum \limits_{\alpha \in \Psi^{-}_w}\alpha,\]
where $\Psi^-_w$ is the set of negative roots of $(\mathfrak{g}, \mathfrak{h})$ which are mapped into the set of positive ones by $w^{-1}$.
\end{lemma}
\begin{proof}
Lemma \ref{lem:weylroots} is classical in Lie theory. For example it was used by Knapp in \cite{Knapp2002} (Lemma 5.112) to give a proof of Weyl character formula. We provide a proof here just for completeness.

Let $\rho = \frac{1}{2}\sum \limits_{\alpha \in \Delta_+}\alpha = \sum \limits_{i = 1}^{\ell}\omega_i$, then for any $w \in \mathcal{W}$ we have
\[-\rho = \frac{1}{2} \sum \{\alpha\ \vert\ \alpha < 0,\ w^{-1}\alpha > 0\} + \frac{1}{2}\sum\{\alpha \ \vert\ \alpha < 0,\ w^{-1}\alpha < 0\}\]
and 
\begin{align*}
w\rho & = -\frac{1}{2}w \sum \{\beta \ \vert\ \beta < 0,\ w\beta > 0\} - \frac{1}{2}w \sum \{\beta \ \vert\ \beta < 0,\ w\beta < 0\}\\
& = -\frac{1}{2}\sum \{w\beta \ \vert\ \beta < 0,\ w\beta > 0\} - \frac{1}{2}\sum \{w\beta \ \vert\ \beta < 0,\ w\beta < 0\}\\
& = -\frac{1}{2}\sum \{\alpha \ \vert\ w^{-1}\alpha < 0,\ \alpha > 0\} - \frac{1}{2}\sum \{\alpha \ \vert\ w^{-1}\alpha < 0,\ \alpha < 0\}\\
& = \frac{1}{2}\sum \{\alpha \ \vert\ w^{-1}\alpha > 0,\ \alpha < 0\} - \frac{1}{2}\sum \{\alpha \ \vert\ w^{-1}\alpha < 0,\ \alpha < 0\}
\end{align*}
Adding these two expressions together, we obtain
\[(w - 1)\rho = \sum \{\alpha \ \vert\ \alpha < 0,\ w^{-1}\alpha > 0\}.\]
\end{proof}

Equipped with Lemma \ref{lem:weylroots}, Theorem \ref{thm:Kmatrix} can be proved as follows.
\begin{proof}[Proof of Theorem \ref{thm:Kmatrix}]
By \eqref{eq:ai0w} we have
\begin{align*}
E^w_{\alpha} & = -L(\alpha) - \sum \limits_{i = 1}^{\ell}\alpha(h_{\alpha_i})a^w_{i0} \\
& = -L(\alpha) - \sum \limits_{i = 1}^{\ell}\alpha(h_{\alpha_i})\sum \limits_{\beta \in \Phi^-_w}\kappa(\omega_i, w^{-1}\check{\beta})\\
& = -L(\alpha) - \sum \limits_{\beta \in \Phi_w^-} \sum \limits_{i = 1}^{\ell} \kappa(\alpha, \check{\alpha}_i)\kappa(\omega_i, w^{-1}\check{\beta}) \\
&= \sum \limits_{i = 1}^{\ell} \kappa(\alpha, \hat{\omega}_i) - \sum \limits_{\beta \in \Phi_w^-}\kappa(\alpha, w^{-1}\check{\beta}) \\
& = \sum \limits_{i = 1}^{\ell} \kappa(\alpha, w^{-1}\hat{\omega}_i) \qquad (\text{by Lemma \ref{lem:weylroots}})\\
& = L(w\alpha), 
\end{align*}
where $\hat{\omega}_i \in \mathfrak{h}^*,\ 1 \le i \le \ell$ are taken so that $\kappa(\alpha_i, \hat{\omega}_j) = \delta_{ij}$. 
\end{proof}


\subsection{Compatibility conditions}
We have seen in Theorem \ref{thm:Onetoone} that $\overline{c(\mathcal{F}_{\Lambda})}$ meets all the Bruhat cells, so that the at worst simple pole singularities in Proposition \ref{prop:Uniqueness} are all produced by various $\tau$-functions $\tau_k(t; w_*), 1 \le k \le \ell$ for $w \in \mathcal{W}$ through equation \eqref{eq:ai}. For each positive integer eigenvalue (counting multiplicity) of $\mathcal{K}_{w}$ we may introduce a free parameter into the Laurent series solutions. Adjoint conditions \eqref{cond:Fredholm} (a Fredholm type condition) for the system \eqref{eq:Higherorderterm} must be checked to ensure that these free parameters are not killed in later iteration steps. Corollary \ref{cor:freeparameters} then states that fixing spectral parameters which is invariant under the f-KT flow, for each $w \in \mathcal{W}$ the maximal number of free parameters in the corresponding Laurent series solutions $a^w_i(t),\ w \in \mathcal{W},\ 1 \le i \le \ell$ matches with the dimension of the Bruhat cell $N_-\dot{w}B_+ \slash B_+$.

The general theory for verifying these adjoint conditions are limited (c.f. \cite{Adler-vanMoerbeke-Vanhaecke2013}, Proposition 7.17, Proposition 7.22 for some general information for algebraically integrable systems), however, in our specific cases these adjoint conditions are automatically satisfied. The reason is as follows. Recall that f-KT lattice written in terms of the entries $a_i$'s and $b_{\alpha}$'s of the Lax matrix $L_{\mathfrak{g}}$ takes the following form
\begin{align}\label{eq:f-KTincoordinate2}
& \frac{d a_i}{d t} = b_i \quad (1 \le i \le \ell), \nonumber\\
& \frac{d b_{\alpha}}{d t} = - \sum_{i = 1}^{\ell}\left(\alpha(H_{\alpha_i})a_i\right)b_{\alpha} + \sum \limits_{i = 1}^{\ell}N_{\alpha_i, -\alpha - \alpha_i}b_{\alpha + \alpha_i}, \quad \alpha \in \Delta_+.
\end{align}

Arguing along the same line as in the proof of Proposition \ref{prop:Uniqueness} and proof of Proposition \ref{conj:Nonnegative} we can express all the $b_{\alpha},\ \alpha \in \Delta_+$ in terms of $a_{i},\ 1 \le i \le \ell$ and obtain $\ell$ higher order differential equations for $a_i(t)$'s, then all the $b_{\alpha}(t)$ can then be subsequently determined. That is all the free parameters in the Laurent series solutions \eqref{eq:Laurentansatz} of f-KT lattice are generated by coefficients $a_{ik},\ k \ge 1$, $1 \le i \le \ell$ of $a_i(t)$'s  (which is equivalent to coefficients of $b_{\alpha_i},\ 1 \le i \le \ell$).

The adjoint conditions that need to be checked arise from the differential equations for various $b_{\alpha+\alpha_i}$'s (i.e. the second part of the system \eqref{eq:f-KTincoordinate2}) as there might exist more than one way to reach $b_{\alpha + \alpha_i}$ from $a_i$ and $b_{\beta}, \beta \in \Delta_+$ through \eqref{eq:f-KTincoordinate2} where $b_{\beta}$'s have lower height than $b_{\alpha + \alpha_i}$. To see that all the adjoint conditions are automatically satisfied, let us consider the centerless affine Lie algebra $\overline{\mathfrak{g}^{(1)}} = \mathbb{C}((t, t^{-1})) \otimes \mathfrak{g} \oplus \mathrm{d}$ where $\mathbb{C}((t)) = \mathbb{C}[t^{-1}] \oplus t\mathbb{C}[[t]]$ and $\mathrm{d}$ is the derivation which acts on $\mathbb{C}((t, t^{-1})) \otimes \mathfrak{g}$ by $\frac{d}{dt}$. Note that the $L_{\mathfrak{g}}(t)$'s can be viewed as elements in $\overline{\mathfrak{g}^{(1)}}$, and the Kostant-Toda lattice \eqref{eq:f-KThky} is nothing but the algebraic relation $[\mathrm{d}, L_{\mathfrak{g}}(t)] = [B(t), L_{\mathfrak{g}}(t)]$ in $\overline{\mathfrak{g}^{(1)}}$. The consistency requirement we need here comes from the basic fact in the structure theory of Lie algebras: for each root $\alpha \in \Delta_+$ the root vector $X_{\alpha}$ can be generated by commutators of root vectors associated with simple roots and is one dimensional, thus independent of the root paths we take to reach it. In the current case under consideration, we similarly have that once $a_i(t)$'s are chosen then all the $b_{\alpha}(t) \otimes Y_{\alpha}$'s are uniquely determined by commutators of $d$ and $a_i(t) \otimes H_i$ no matter which root path we choose (c.f. \cite{Adler-vanMoerbeke1982} for a partial inverse of the results here on the tridiagonal periodic Toda lattice). 

Thus we have
\begin{corollary}\label{cor:compatibility}
For the Laurent series solutions \eqref{eq:Laurentansatz} associated with $w \in \mathcal{W}$, there are exactly $\vert \Lambda_w \vert = \ell + l(w_0) - l(w)$ many independent free parameters in it.
\end{corollary}

At last, by Proposition \ref{prop:Uniqueness} all the formal Laurent series solutions we constructed are weight homogeneous, then by Theorem \ref{thm:K-matrixgeneral} we know that they converge and indeed are meromorphic solutions of f-KT lattice. 

\subsection{f-KT lattice and flag variety}
In summary we have
\begin{theorem}\label{thm:generaltau}
For any $\Lambda \in \mathbb{C}^{\ell}$, $w \in \mathcal{W}$ and $u \in N_-$
\begin{equation}\label{eq:tau2}
\tau_{k}(\t; w) = (v^{\omega_k}, \exp(\Theta_{C_{\Lambda}}(\t))u\dot{w}v^{\omega_k}), \qquad 1 \le k \le \ell,
\end{equation}
are $\tau$-functions of f-KT lattice. The diagonal elements are given by
\begin{equation}\label{eq:ai2}
a_k(\t) = \frac{d}{dt}\ln \tau_k(\t),  \qquad 1 \le k \le \ell,
\end{equation}
and all the other entries $b_{\alpha}(\t), \alpha \in \Delta_+$ in the Lax matrix $L_{\mathfrak{g}}(\t)$ can be subsequently determined from the $a_k(\t)$'s.
\end{theorem}

As function of $t$, $\tau$-functions as given in \eqref{eq:tau2} are holomorphic functions in the complex domain, so their zeroes are isolated. From this fact we deduce
\begin{corollary}
If the Toda flow $\exp (tC_{\Lambda})uB_+$ hits a smaller Bruhat cell, then it intersects with that Bruhat cell transversely.
\end{corollary}


To further illustrate the relation between Laurent series solutions of the f-KT lattice and the corresponding $\tau$-functions, we would like to demonstrate how the free parameters in $a_i(t)$ and the type ``X'' eigenvalues of the Kowalevski matrices arise from the $\tau$-functions defined in \eqref{eq:tau2}.

We first note the following easy lemma.
\begin{lemma}\label{eq:auxiliarylemma}
Suppose $x(t) = \frac{d}{dt}\log y(t)$, and $y(t) = \sum \limits_{i = 0}^{\infty}b_it^{i+k}$, where $b_0 = 1,\ k \in \mathbb{Z}_{+}$. Assume that the coefficients of $y(t)$ are generated by $n$ independent free parameters, i.e. $b_i \in \mathbb{C}[b_{\alpha_1}, b_{\alpha_2}, \dots, b_{\alpha_n}]$ for all integers $i \ge 0$. Then a basis of free parameters appearing in the coefficient of the Laurent series of $x(t) = \sum \limits_{j = 0}^{\infty}a_jt^{j - 1}$ ($a_{0} = k$) can be chosen as $\{a_{\alpha_1}, a_{\alpha_2}, \dots, a_{\alpha_n}\}$.
\end{lemma}
\begin{proof}
Coefficients of $x(t)$ can be recursively obtained from the coefficients of $y(t)$ through the relation $y'(t) = x(t)y(t)$. More precisely, the lemma follows from the following recursive relations
\[ \sum \limits_{i = 0}^{j}a_{j-i}b_i = (j+k)b_{j}, \qquad j \ge 0,\]
where $b_0 = 1$ and $a_{0} = k$.
\end{proof}

\begin{proof}[Proof of Theorem \ref{thm:generaltau}]
By Theorem \ref{thm:Kmatrix} and Corollary \ref{cor:compatibility}, we know exactly where the free parameters arise in the Laurent series solutions. To prove Theorem \ref{thm:generaltau} we match these parameters with free parameters appearing in $\tau$-functions \eqref{eq:tau2}. 

Note that the general $\tau$-functions \eqref{eq:tau2} can be written as
\begin{equation}
\tau_i(t; w) = (v^{\omega_i}, e^{tC_{\Lambda}}u\dot{w}v^{\omega_i}) = (v^{\omega_i}, e^{tC_{\Lambda}}e^{U}\dot{w} v^{\omega_i}) \qquad 1 \le i \le \ell, \ w \in \mathcal{W}, \ u = e^{U} \in N_-.
\end{equation}
As $U \in \mathfrak{n}_-$ and $C_{\Lambda} \in e + \mathfrak{s} \subseteq e + \mathfrak{n}_-$ ($\ell$-dimensional), we can write $U = \sum_{\alpha \in \Delta_+}U_{\alpha}Y_{\alpha}$ and $C_{\Lambda} = e + \sum_{i = 1}^{\ell}C_{i}Y_{\beta_i}, \beta_i \in \Delta_+$.
Note that as functions of $U_{\alpha}, C_{i}$ and $t$, $\tau_i(t; w)$'s are weight homogeneous functions if we assign weights $\varpi(U_{\alpha}) = -L(\alpha)$, $\varpi(C_{i}) = - L(\beta_i)$ and $\varpi(t) = 1$, that is weights of $U_{\alpha}, C_{i}$ are given by the heights of the corresponding roots. 

Writing the $\tau$-functions in its Taylor expansion form with respect to $t$:
\[\tau_i(t; w) = \sum \limits_{j = 0}^{\infty}\tau_{ij}t^{j + k_i} \qquad 1 \le i \le \ell,\]
where $\tau_{i0} = 1$ and $k_i$ depending on $w \in \mathcal{W}$ is the leading exponent of $\tau_i(t; w)$, then we have
\begin{enumerate}
\item with the assignned weights for $U_{\alpha}, C_{i}$ and $t$, $\tau_i(t; w)$ is weighted homogeneous with weight $k_i(w)$,
\item the free parameters $\tau_{ij}$ in the Taylor expansion of $\tau_i(t; w)$'s come from $U_{\alpha}$ and $C_{i}$, $\alpha \in \Delta_+$,
\item let $I \subseteq \Delta_+$ be such that $\{U_{\alpha}, \alpha \in I; C_i, 1 \le i \le \ell\}$ is the set of minimum generators of free parameters in $\tau(t; w)$, equivalently this means $\{U_{\alpha}, \alpha \in I\}$ is a coordinate system for the Bruhat cell $N_-\dot{w}B_+$, then a basis of free parameters in the Taylor expansion of $\tau$-functions can be taken as subset of $\{\tau_{ik}, 1 \le i \le \ell, k = L(\alpha), \alpha \in I\text{ or } k = L(\beta_j), 1 \le j \le \ell\}$ (the cardinality of the generators is $\ell + \vert I \vert$).
\end{enumerate}
By Lemma \ref{eq:auxiliarylemma} we thus have matched all the free parameters in Laurent series solutions with free parameters in $\tau$-functions \eqref{eq:tau2} and finish the proof of Theorem \ref{thm:generaltau}.
\end{proof}

Finally we have the following global structure for all of the Laurent series solutions of f-KT lattice.
\begin{theorem}\label{conj:Flag}
For any $\Lambda \in \mathbb{C}^{\ell}$, the compactification of $c_{\Lambda}(F_{\Lambda})$ is $G \slash B_+$. All the Laurent series solutions of f-KT hierarchy are parameterized by ${{G} \slash {B}_+ \times \mathbb{C}^{\ell}}$, where ${G} \slash {B}_+$ is the flag variety and ${ \mathbb{C}^{\ell}}$ parametrizes the data for spectral parameters. 
\end{theorem}

\begin{proof}
The Laurent series solutions associated with $w \in \mathcal{W}$ and the $\tau$-functions associated with the same $w \in \mathcal{W}$ have the same singularity structure and contain the same number of free parameters at the same positions, thus they give the same solution for the f-KT lattice. 
\end{proof}

\begin{remark}
\begin{enumerate}
\item There should exist an explicit correspondence between Bruhat decomposition of the flag variety here and the Bruhat order in the study of the full symmetric Toda lattice (c.f. \cite{Chernyakov-Sharygin-Sorin2012}). 
\item The Laurent series solutions studied in this paper actually contain solutions for various intermediate Kostant-Toda lattices (c.f. \cite{Flaschka-Haine1991, Damianou-Sabourin-Vanhaecke2015}). For example, Laurent series solutions corresponding to $w_0$ will force $b_{\alpha} = 0$ for $\alpha \not\in \Pi$, so they give solutions to the classical tridiagonal Kostant-Toda lattice. The other solutions of tridiagonal Kostant-Toda lattice can be obtained by manipulating the free parameters so that $b_{\alpha} = 0$ for $\alpha \not\in \Pi$ (c.f. Section \ref{sec:Painleveso5} for an explicit example). So in some sense the Laurent series solutions can see all of the coadjoint orbits simultaneously. 
\end{enumerate}
\end{remark}

\begin{corollary}
Only f-KT flows belong to solutions of the tridiagonal Kostant-Toda lattice intersect with the smallest dimensional Bruhat cell.
\end{corollary}

\section{An example:f-KT lattice on $\mathfrak{so}_5(\mathbb{C})$}\label{sec:Painleveso5}
In this section we consider f-KT lattice on the rank $2$ type $B$ Lie algebra as an example to illustrate the main results in the previous sections. 

Let $V$ be a $5$ dimensional vector space, and $\Psi$ a non-degenerate symmetric bilinear form on $V$ of maximal index $2$. Then the quadratic form $Q$ associated to $\Psi$ is $Q(x) = \frac{1}{2}\Psi(x, x)$ for $x \in V$. We take a basis $\Xi = (e_{-2}, e_{-1}, e_0, e_{1}, e_{2})$ of $V$ such that $\Psi(e_{i}, e_{j}) = (-1)^i\delta_{i, -j}$, then in particular we have $\Psi(e_0, e_0) = 1$ and
\[Q(\sum x_ie_i) = \frac{1}{2}x_0^2 + \sum \limits_{i = 1}^2 (-1)^ix_ix_{-i}.\]
The (Gram) matrix of $\Psi$ with respect to basis $\Xi$ is
\[M_2 = \begin{pmatrix}
0 & 0 & 0 & 0 & 1\\
0 & 0 & 0 & -1 & 0\\
0 & 0 & 1 & 0 & 0\\
0 & -1 & 0 & 0 & 0\\
1 & 0 & 0 & 0 & 0
\end{pmatrix}.\]
The Lie algebra $\mathfrak{so}_5(\mathbb{C})$ we consider here is the orthogonal Lie algebra associated with $(V, \Psi)$.

The f-KT lattice on $\mathfrak{so}_5(\mathbb{C})$ is defined as
\[\frac{dL}{dt} = [B, L],\]
where
\[L = \begin{pmatrix}
a_2 & 1 & 0 & 0 & 0\\
b_2 & 2a_1 - a_2 & 1 & 0 & 0\\
2c_1 & 2b_1 & 0 & 1 & 0\\
4d_1 & 0 & 2b_1 & a_2 - 2a_1 & 1\\
0 & 4d_1 & -2c_1 & b_2 & -a_2
\end{pmatrix} \quad \text{and} \quad B = \begin{pmatrix}
a_2 & 1 & 0 & 0 & 0\\
0 & 2a_1 - a_2 & 1 & 0 & 0\\
0 & 0 & 0 & 1 & 0\\
0 & 0 & 0 & a_2 - 2a_1 & 1\\
0 & 0 & 0 & 0 & -a_2
\end{pmatrix}.\]
More explicitly, we have
\begin{align*}
& \frac{d}{dt}a_2 = b_2 \qquad \frac{d}{dt}a_1 = b_1\\
& \frac{d}{dt}b_2 = (2a_1 - 2a_2)b_2 + 2c_1 \qquad  \frac{d}{dt}b_1 = (a_2 - 2a_1)b_1 - c_1\\
& \frac{d}{dt}c_1 = -a_2c_1 + 2d_1 \qquad \frac{d}{dt}d_1 = -2a_1d_1.
\end{align*}

It can be checked directly that the Laurent series solutions have the following form
\begin{align*}
& a_i(t) = \sum \limits_{k = 0}^{\infty}a_{ik}t^{k-1}, \qquad b_i(t) = \sum \limits_{k = 0}^{\infty}b_{ik}t^{k-2},\\
& c_1(t) = \sum \limits_{k = 0}^{\infty}c_{1k}t^{k-3}, \qquad d_1(t) = \sum \limits_{k = 0}^{\infty}d_{1k}t^{k-4}.
\end{align*}
Substituting them into the above differential equations, we get
\begin{align*}
& \sum \limits_{k = 0}^{\infty}(k-1)a_{2k}t^{k-2} = \sum \limits_{k = 0}^{\infty}b_{2k}t^{k-2}, \qquad  \sum \limits_{k = 0}^{\infty}(k-1)a_{1k}t^{k-2} = \sum \limits_{k = 0}^{\infty}b_{1k}t^{k-2},\\
& \sum \limits_{k = 0}^{\infty}(k - 2)b_{2k}t^{k-3} = 2[\sum \limits_{i = 0}^{\infty}(a_{1i} - a_{2i})t^{i-1}][\sum \limits_{j = 0}^{\infty}b_{2j}t^{j-2}] + 2\sum \limits_{k = 0}^{\infty}c_{1k}t^{k-3},\\
& \sum \limits_{k = 0}^{\infty}(k - 2)b_{1k}t^{k-3} = [\sum \limits_{i = 0}^{\infty}(a_{2i} - 2a_{1i})t^{i-1}][\sum \limits_{j = 0}^{\infty}b_{1j}t^{j-2}] - \sum \limits_{k = 0}^{\infty}c_{1k}t^{k-3},\\
& \sum \limits_{k = 0}^{\infty}(k - 3)c_{1k}t^{k - 4} = -(\sum \limits_{i = 0}^{\infty}a_{2i}t^{i-1})(\sum \limits_{j = 0}^{\infty}c_{1j}t^{j-3}) + 2 \sum \limits_{k = 0}^{\infty}d_{1k}t^{k-4},\\
& \sum \limits_{k = 0}^{\infty}(k - 4)d_{1k}t^{k-5} = -2(\sum \limits_{i = 0}^{\infty}a_{1i}t^{i-1})(\sum \limits_{j = 0}^{\infty}d_{1j}t^{j-4}).
\end{align*}

Comparing the power of $t$ for $k = 0$, we obtain the following indicial equations for the leading coefficients
\begin{equation}\label{eq:IndicialB2}
\left\{
\begin{array}{l}
 -a_{20} = b_{20}, \qquad -a_{10} = b_{10},\\
 -2b_{20} = 2(a_{10} - a_{20})b_{20} + 2c_{10},\\
 -2b_{10} = (a_{20} - 2a_{10})b_{10} - c_{10},\\
 -3c_{10} = -a_{20}c_{10} + 2d_{10},\\
 -4d_{10} = -2a_{10}d_{10}.
\end{array}
\right.
\end{equation}
Note that $c_{10}, d_{10}$ are determined once we know $a_{20}$ and $a_{10}$. Eliminating $c_{10}$ from the third and fourth equations in \eqref{eq:IndicialB2}, we get a degree $2$ equation for $a_{20}$ and $a_{10}$. Substituting either of the expressions of $c_{10}$ from the third or fourth equation to the fifth one we get an expression for $d_{10}$ and substituting this expression to the last equation we get a degree $4$ equation for $a_{10}$ and $a_{20}$. So the maximal number of solutions is $8 = \vert \mathcal{W} \vert$.

Let $\mathfrak{h}$ be the set of diagonal elements of $\mathfrak{so}_5(\mathbb{C})$. This is a commutative subalgebra of $\mathfrak{so}_5(\mathbb{C})$. Let 
\[H_i = E_{-i, -i} - E_{i, i}, \qquad (1 \le i \le 2).\]
Let $(\varepsilon_i)_{1 \le i \le 2}$ be the basis of $\mathfrak{h}^*$ dual to $(H_i)$ such that $\varepsilon_i(E_{-j, -j} - E_{j, j}) = \delta_{ij}$.
Then a basis of $\mathfrak{h}$ can be taken as
\[H_{\varepsilon_2 - \varepsilon_1} = H_2 - H_1, \qquad H_{\varepsilon_1} = 2H_{1}.\]
The positive roots of $\mathfrak{so}_5(\mathbb{C})$ are 
\[\alpha_2 = \varepsilon_2 - \varepsilon_1, \quad \alpha_1 = \varepsilon_1, \quad \alpha_3 = \alpha_1 + \alpha_2 = \varepsilon_2, \quad \alpha_4 = 2\alpha_1 + \alpha_2 = \varepsilon_1 + \varepsilon_2,\]
and the coroots are
\[\check{\alpha}_2 = \alpha_2, \quad \check{\alpha}_1 = 2\alpha_1, \quad \check{\alpha}_3 = 2(\alpha_1 + \alpha_2) = \check{\alpha}_1 + 2\check{\alpha}_2, \quad \check{\alpha}_4 = 2\alpha_1 + \alpha_2 = \check{\alpha}_1 + \check{\alpha}_2. \]
The Weyl group elements are given by
\[\mathcal{W}_{B_2} = \left\{\begin{array}{c}
e, s_2^B, s_1^B, w_{21} = s_2^Bs_1^B, w_{12} = s_1^Bs_2^B, w_{212} = s_2^Bs_1^Bs_2^B, \\
w_{121} = s_1^Bs_2^Bs_1^B, s_2^Bs_1^Bs_2^Bs_1^B = w_0
\end{array}\right\},\]
where $s_2^B = s_{-2}s_1, s_1^B = s_{-1}s_0s_{-1}$ with $s_i$ the simple reflection $(i, i+1)$.
We have

\begin{table}[h]
\begin{center}
\begin{minipage}{0.56\linewidth}
\caption{Weyl group action on coroots}\label{tabcroots}%
\begin{tabular}{|c|c|c|c|c|}
\hline
& $\check{\alpha}_2$ & $\check{\alpha}_1$  & $2\check{\alpha}_2 + \check{\alpha}_1$ & $\check{\alpha}_2 + \check{\alpha}_1$\\
\hline
$e$  &  $\check{\alpha}_2$ & $\check{\alpha}_1$  & $2\check{\alpha}_2 + \check{\alpha}_1$ & $\check{\alpha}_2 + \check{\alpha}_1$\\
\hline
$s_1^B$   & $\check{\alpha}_2 + \check{\alpha}_1$   & $-\check{\alpha}_1$ & $2\check{\alpha}_2 + \check{\alpha}_1$ & $\check{\alpha}_2$\\
\hline
$s_2^B$    & $-\check{\alpha}_2$  & $2\check{\alpha}_2 + \check{\alpha}_1$  & $\check{\alpha}_1$  & $\check{\alpha}_2 + \check{\alpha}_1$\\
\hline
$w_{21}$    & $\check{\alpha}_2 + \check{\alpha}_1$   & $-(2\check{\alpha}_2 + \check{\alpha}_1)$  & $\check{\alpha}_1$  & $-\check{\alpha}_2$\\
\hline
$w_{12}$    & $-(\check{\alpha}_2 + \check{\alpha}_1)$   & $2\check{\alpha}_2 + \check{\alpha}_1$  & $-\check{\alpha}_1$  & $\check{\alpha}_2$\\
\hline
$w_{121}$    & $\check{\alpha}_2$   & $-(2\check{\alpha}_2 + \check{\alpha}_1)$  & $-\check{\alpha}_1$  & $-(\check{\alpha}_2 + \check{\alpha}_1)$\\
\hline
$w_{212}$    & $-(\check{\alpha}_2 + \check{\alpha}_1)$   & $\check{\alpha}_1$  & $-(2\check{\alpha}_2 + \check{\alpha}_1)$  & $-\check{\alpha}_2$\\
\hline
$w_0$    & $-\check{\alpha}_2$   & $-\check{\alpha}_1$  & $-(2\check{\alpha}_2 + \check{\alpha}_1)$  & $-(\check{\alpha}_2 + \check{\alpha}_1)$\\
\hline
\end{tabular}
\end{minipage}
\end{center}
\end{table}

Thus we have
\begin{align*}
& \sum \limits_{\check{\alpha} \in \Phi_{e}}\check{\alpha} = 0; \quad 
 \sum \limits_{\check{\alpha} \in \Phi_{s^B_{1}}}\check{\alpha} = \check{\alpha}_1; \quad
 \sum \limits_{\check{\alpha} \in \Phi_{s^B_{2}}}\check{\alpha} = \check{\alpha}_2; \\
 & \sum \limits_{\check{\alpha} \in \Phi_{w_{21}}}\check{\alpha} = \check{\alpha}_2 + 2\check{\alpha}_1; \quad
 \sum \limits_{\check{\alpha} \in \Phi_{w_{12}}}\check{\alpha} = 3\check{\alpha}_2 + \check{\alpha}_1; \quad
\sum \limits_{\check{\alpha} \in \Phi_{w_{121}}}\check{\alpha} = 3\check{\alpha}_2 + 3\check{\alpha}_1; \\
& \sum \limits_{\check{\alpha} \in \Phi_{w_{212}}}\check{\alpha} = 4\check{\alpha}_2 + 2\check{\alpha}_1; \quad
 \sum \limits_{\check{\alpha} \in \Phi_{w_{0}}}\check{\alpha} = 4\check{\alpha}_2 + 3\check{\alpha}_1.
\end{align*}
Therefore according to Theorem \ref{thm:Onetoone} we have the following $8$ solutions for the indicial equations \eqref{eq:IndicialB2} which can also be checked directly

\begin{table}[h]
\begin{center}
\begin{minipage}{0.44\linewidth}
\caption{Solutions of the indicial equations}\label{tabindicial}%
\begin{tabular}{|c|c|c|c|c|c|c|}
\hline
Level.case & $a_{20}$ & $a_{10}$  & $b_{20}$ & $b_{10}$ & $c_{10}$ & $d_{10}$\\
\hline
0.1  &  $0$ & $0$  & $0$ & $0$ & $0$ & $0$\\
\hline
1.1   & $0$   & $1$ & $0$ & $-1$ & $0$ & $0$\\
\hline
1.2    & $1$  & $0$  & $-1$  & $0$ & $0$ & $0$\\
\hline
2.1    & $1$   & $2$  & $-1$  & $-2$ & $2$ & $-2$\\
\hline
2.2    & $3$   & $1$  & $-3$  & $-1$ & $-3$ & $0$\\
\hline
3.1    & $3$   & $3$  & $-3$  & $-3$ & $3$ & $0$\\
\hline
3.2    & $4$   & $2$  & $-4$  & $-2$ & $-4$ & $-2$\\
\hline
4.1    & $4$   & $3$  & $-4$  & $-3$ & $0$ & $0$\\
\hline
\end{tabular}
\end{minipage}
\end{center}
\end{table}


For $k \ge 1$, the coefficients of the Laurent series solutions can be found by the following iterative procedure:
\begin{align}\label{eq:recursiveso5}
& (k-1)a_{2k} - b_{2k} = 0, \qquad (k-1)a_{1k} - b_{1k} = 0, \nonumber\\
& (k - 2 - 2a_{10} + 2a_{20}) b_{2k} + 2b_{20}a_{2k} - 2b_{20}a_{1k} - 2c_{1k} = 2\sum \limits_{i = 1}^{k-1}b_{2i}(a_{1, k-i} - a_{2, k-i}),\nonumber\\
& (k - 2 - a_{20} + 2a_{10}) b_{1k} - b_{10}a_{2k} + 2b_{10}a_{1k} + c_{1k} = \sum \limits_{i = 1}^{k -1}b_{1i}(a_{2, k-i} - 2a_{1, k-i}),\\
& (k - 3 + a_{20})c_{1k} + c_{10}a_{2k} - 2d_{1k} = -\sum \limits_{i = 1}^{k - 1}a_{2i}c_{1, k-i},\nonumber\\
& (k - 4 + 2a_{10})d_{1k} + 2d_{10}a_{1k} = -2 \sum \limits_{i = 1}^{k -1}a_{1i}d_{1, k-i}. \nonumber
\end{align}
The coefficient matrix $kI - \mathcal{K}$ in the order $(a_{2k}, a_{1k}, b_{2k}, b_{1k}, c_{1k}, d_{1k})$ has the following form
\begin{align*}
\begin{pmatrix}
k-1 & 0 & -1 & 0 & 0 & 0\\
0 & k-1 & 0 & -1 & 0 & 0\\
2b_{20} & -2b_{20} & k-2-2a_{10} + 2a_{20} & 0 & -2 & 0\\
-b_{10} & 2b_{10} & 0 & k-2-a_{20} + 2a_{10} & 1 & 0\\
c_{10} & 0 & 0 & 0 & k-3+a_{20} & -2\\
0 & 2d_{10} & 0 & 0 & 0 & k-4+2a_{10}
\end{pmatrix}.
\end{align*}

\begin{lemma}
The determinant of $kI - \mathcal{K}$ can be calculated directly as (by taking the indicial equations \eqref{eq:IndicialB2} into account)
\[\det(kI - \mathcal{K}) = (k-2)(k-4)(k-1+2a_{20}-2a_{10})(k-1-a_{20}+2a_{10})(k-2+a_{20})(k-3+2a_{10}),\]
so again the $k$'s where $\mathcal{K}$ degenerates are all integers. Note that $2, 4$ are degrees of the Chevalley invariants, and the other $4$ are eigenvalues of X-type.
Let $E_{\alpha_2}:=1-2a_{20}+2a_{10}$, $E_{\alpha_1}:= 1+a_{20} - 2a_{10}$, $E_{\alpha_3}:= 2-a_{20}$, $E_{\alpha_4}:=3-2a_{10}$, then we can make the following table (see Table \ref{tabspectra}).

\begin{table}[h]
\begin{minipage}{0.7\linewidth}
\caption{Spectra of the Kowalevski matrix}\label{tabspectra}%
\begin{tabular}{|c|c|c|c|c|c|c|c|}
\hline
Level.case & $E_{\alpha_2}$ & $E_{\alpha_1}$  & $E_{\alpha_3}$ & $E_{\alpha_4}$ & $\mathcal{W}$ & $l(w)$ & Number of free parameters\\
\hline
0.1  &  $1$ & $1$  & $2$ & $3$ & $e$ & $0$ & 6\\
\hline
1.1   & $3$   & $-1$ & $2$ & $1$ & $s_1^B$ & $1$ & 5\\
\hline
1.2    & $-1$  & $2$  & $1$  & $3$ & $s_2^B$ & $1$ & 5\\
\hline
2.1    & $3$   & $-2$  & $1$  & $-1$ & $w_{21}$ & $2$ & 4\\
\hline
2.2    & $-3$   & $2$  & $-1$  & $1$ & $w_{12}$ & $2$ & 4\\
\hline
3.1    & $1$   & $-2$  & $-1$  & $-3$ & $w_{121}$ & $3$ & 3\\
\hline
3.2    & $-3$   & $1$  & $-2$  & $-1$ & $w_{212}$ & $3$ & 3\\
\hline
4.1    & $-1$   & $-1$  & $-2$  & $-3$ & $w_0$ & $4$ & 2\\
\hline
\end{tabular}
\end{minipage}
\end{table}

\end{lemma}
Here we can see explicitly that the Weyl group action on $E_{\alpha}$ is equivariant to the Weyl group action on the root system of $\mathfrak{so}_5(\mathbb{C})$.

\begin{remark}
From Table \ref{tabindicial} and Equations \eqref{eq:recursiveso5} we can also easily read all the sub-hierarchies of the f-KT hierarchy. 
\begin{enumerate}
\item $a_2(t) = b_2(t) = c_1(t) = d_1(t) \equiv 0$, and only $a_1(t)$ and $b_1(t)$ are nontrivial. In this case, $a_1(t)$ and $b_1(t)$ can be explicitly solved and the flow can only enter into the smaller Bruhat cell associated with $s_1$.
\item $a_1(t) = b_1(t) = c_1(t) = d_1(t) \equiv 0$, and only $a_2(t)$ and $b_2(t)$ are nontrivial. Similarly, $a_2(t)$ and $b_2(t)$ can be explicitly solved and the flow can only enter into the smaller Bruhat cell associated with $s_2$.
\item $c_1(t) = d_1(t) \equiv 0$, and $a_2(t), a_1(t), b_2(t), b_1(t)$ are nontrivial. This is the classical tridiagonal Kostant-Toda lattice, and the flow can enter into the smaller Bruhat cells associated with $s_1, s_2$ and $w_{0}$ which are cases 1.1, 1.2 and 4.1 respectively. This is consistent with the results in \cite{Flaschka-Haine1991}. Note also that in case 4.1, by \eqref{eq:recursiveso5} $a_{20} = 4, a_{10} = 3$ forces $c_1(t) = d_1(t) \equiv 0$, which means that only flows of the classical tridiagonal Kostant-Toda lattice can enter into the smallest dimensional Bruhat cell. This counter-intuitive result is not easily deduced from the $\tau$-function formalism.  
\item $d_1(t) \equiv 0$, and all the other matrix elements are nontrivial. This is the so-called $2$-banded Kostant-Toda lattice (c.f. \cite{Kodama-Xie2021f-KT, Xie2021}), and the flows can enter into all the smaller Bruhat cells except the ones associated with $w_{21}$ and $w_{212}$ which are cases 2.1 and 3.2 respectively. Note again that in case 3.1, $d_1(t)$ is forced to be $0$ which is different from case 2.2 where $d_1(t)$ is set to be $0$ by letting $d_{12} = 0$. That is only the flows of the $2$-banded Kostant-Toda lattice can enter into the Bruhat cell associated with $w_{121}$.
\end{enumerate}
\end{remark}

At last we comment that for the tridiagonal Kostant-Toda lattice, that is in cases 0.1, 1.1, 1.2 and 4.1, the spectra of the Kowalevski matrices are $(1, 1, 2, 2)$, $(-1, 1, 2, 4)$, $(-1, 1, 2, 3)$ and $(-3, -1, 2, 4)$ respectively, which do not have simple expressions as in Theorem \ref{thm:Kmatrix}. Therefore, in some sense it is simpler and more fruitful to study f-KT lattice and all the coadjoint orbits of $B_+$ on $\mathfrak{b}_+^*$ together.


\section{Further problems}\label{sec:problems}

In this last section, we formulate several problems based on results in this paper to be considered in the future.

\subsection{Free parameters as canonical coordinates on the flag varieties}

When $a_i(t)$'s, $b_{\alpha}(t)$'s are viewed as entries of the Lax matrix, the free parameters in the Laurent series solutions \eqref{eq:Laurentansatz} can be viewed as functionals on the space of all Lax matrices. According to Theorem \ref{conj:Flag}, these free parameters also provide a natural set of coordinates for the corresponding flag varieties.  As these free parameters arise from a dynamical system in a rather natural manner, we have reasons to believe these coordinates enjoy some good dynamical properties. For example, the Poisson brackets among them as functionals on the space of Lax matrices should have a particularly nice form which could be useful as a new method in constructing constants of motion.
It had been noticed since 1980s (c.f. \cite{Flaschka1988}) that for the tridiagonal Toda lattice, the free parameters in the most degenerate Laurent series solution are determined by the constants of motion. We expect similar things happening in the f-KT lattice case. 

The advantage of these free parameters as constants of motion compared with the chopping method (c.f. \cite{Ercolani-Flaschka-Singer1993, Gekhtman-Shapiro1999}) is as follows: To carry out the chopping method, we need to choose a particular realization of the f-KT lattice in a representation of the Lie algebra in the Lax form \eqref{eq:f-KThky}. For example, as we have $\mathfrak{sl}_4(\mathbb{C}) \cong \mathfrak{so}_6(\mathbb{C})$, the chopping method provides us two sets of mutually non-Poisson commutative constants of motion (c.f. \cite{Ercolani-Flaschka-Singer1993, Shipman2000}). The system \eqref{eq:f-KTincoordinate} however only involves the information from the Cartan matrix which is intrinsic to the corresponding Lie algebra. That is, both sets of constants of motion for $\mathfrak{sl}_4(\mathbb{C}) \cong \mathfrak{so}_6(\mathbb{C})$ f-KT lattice should be expressed in terms of the free parameters in the Laurent series solutions. 

When we try to solve the system \eqref{eq:higerordertermintro} at the resonant steps $k = E_{\alpha}^w > 0$, we have the freedom to choose which $a_{ik}$'s or $b_{\alpha k}$'s we would like to choose as free parameters, and this provides us differential coordinate systems for certain charts of the flag varieties. It is then interesting to study the transitions from one charts to another. According to Lemma \ref{eq:auxiliarylemma} and the generalized Chamber Ansatz (c.f. Theorem 7.1, \cite{Marsh-Rietsch2004}), we further see that these free parameters and therefore all constants of motion can be expressed as rational functions of the generalized minors of $\exp (\Theta_{L_0}(\mathbf{t}))$. It would be nice to relate the above transition of charts to the cluster mutation on the flag varieties (c.f. \cite{Berenstein-Fomin-Zelevinsky2005}).

\subsection{Real solutions of the f-KT lattice}
Now we come back to consider f-KT lattice defined  on a split simple Lie algebra $(\mathfrak{g}, \mathfrak{h})$ over $\mathbb{R}$.
Note that the proof of Proposition \ref{prop:RHfactorization} goes through for all $g \in G_0$, and $\tau_i(\t)$'s in \eqref{eq:tau0} are well defined and are real analytic if $L_{\0}$ is real. Since solutions of the indicial equations \eqref{eq:Indicialf-KT} and the spectra of the Kowalevski matrices \eqref{eq:Xexponent2} are all integers, the Laurent series solutions \eqref{eq:Laurentansatz} are all real valued if all the free parameters introduced into the Laurent series are real. Thus real solutions of f-KT hierarchy may still have all kinds of singular behavior as in the complex cases. A more meaningful question is the following.
\begin{problem}
For a given $L_{\0}$, which Bruhat cells the f-KT flows may hit? And what are the large time asymptotical behaviors for the corresponding real solutions of f-KT hierarchy? 
\end{problem}

Since we now understand all possible types of singular solutions of f-KT hierarchy, the first question is an analogue of the classical connection problem (c.f. \cite{Whittaker-Watson2020} for solution of the connection problem of hypergeometric series). At one of the extremes, the real regular solutions of Toda flows and their asymptotic behaviors especially in type $A$ have been intensively studied (c.f. \cite{Kostant1979a, Gekhtman-Shapiro1997, Kodama-Williams2015}). It is known to Kostant that in the tridiagonal case the Kostant-Toda flows are complete when all the $b_{\alpha_i}, 1 \le i \le \ell$ are positive which is associated with a special element in $\tilde{M}$ in Section \ref{sec:structure}. In this case the Kostant-Toda lattice is conjugate to the classical tridiagonal symmetric Toda lattice whose phase space is compact. The cases when $b_{\alpha_i}$'s have indefinite signs and the solutions blow up in finite time were studied by Casian, Kodama and Ye (c.f. \cite{Kodama-Ye1996b, Casian-Kodama2006}). Gekhtman and Shapiro later pointed out that the regularity condition Kostant give is only sufficient but not necessary (c.f. \cite{Gekhtman-Shapiro1997}). 

For tridiagonal Kostant-Toda lattice defined on $\mathfrak{sl}_{\ell + 1}(\mathbb{R})$, let $\lambda_1, \dots, \lambda_s$ be distinct eigenvalues of $L_0 := L(0)$. Assume that to each eigenvalue $\lambda_i$ there corresponds exactly one Jordan box and let $v_1, \dots, v_{\ell + 1}$ the corresponding Jordan basis. Denote by $V(L_0)$ the matrix with columns $v_1, \dots, v_{\ell + 1}$, and assume $V(L(0))$ has the following LU-factorization
\[V(L_0) = n_{L_0}b_{L_0}, \qquad n_{L_0} \in N_-, b_{L_0} \in B_+.\]
Then we have
\begin{theorem}[\cite{Gekhtman-Shapiro1997}]\label{thm:regular}
The solution of the Kostant-Toda lattice is nonsingular if and only if
\begin{enumerate}
\item all $\lambda_i\, (i = 0, \dots, s)$ are real such that $\lambda_1 > \dots > \lambda_s$,
\item the matrix $n_{L_0}$ is totally positive, that is all its nontrivial minors are positive.
\end{enumerate}
\end{theorem}
Theorem \ref{thm:regular} is also valid for the f-KT lattice. Note that since all the conditions in Theorem \ref{thm:regular} are open conditions, the maximal number of parameters in the generic real regular solutions for f-KT lattice on $\mathfrak{sl}_{\ell + 1}$ is dim $\mathfrak{n}_- + \ell = \ell(\ell + 1) \slash 2 + \ell$. These solutions do not hit any lower dimensional Bruhat cells, which in particular shows that the seemingly obvious conclusion of Theorem \ref{thm:generaltau} is not that trivial.

We also comment that the asymptotic behavior of real regular solutions for tridiagonal Kostant-Toda lattice had been studied by Kostant in \cite{Kostant1979a} which generalized Moser's result \cite{Moser1975a} in type $A$. For f-KT hierarchy in type $A$, the asymptotic behaviors of real regular solutions are described by the so-called Bruhat interval polytopes via the moment map techniques (c.f. \cite{Kodama-Williams2015}). Not much is known for f-KT hierarchy on other type Lie algebras at the moment (c.f. \cite{Xie2021} for some lower rank examples in type $B$).



%


\bibliographystyle{plain}


\end{document}